\documentclass[10pt,a4paper]{article}
\usepackage{amssymb}
\usepackage{amsmath}
\usepackage{amsthm}
\usepackage{latexsym}
\usepackage[dvips]{epsfig}
\usepackage{enumerate}
\usepackage{mathrsfs}
\usepackage{eufrak}
\usepackage{bm}
\usepackage{tikz}
\usepackage{authblk}
\usepackage{cleveref}
\usepackage{cancel}
\usepackage{cancel}

\theoremstyle{plain}
\newtheorem{proposition}{Proposition}
\newtheorem{lemma}{Lemma}
\newtheorem{theorem}{Theorem}

\newtheorem{definition}{Definition}
\newtheorem{remark}{Remark}

\def\bmg{{\bm g}}

\def\bmi{{\bm i}}
\def\bmj{{\bm j}}

\def\bmM{{\bm M}}

\def\bmP{{\bm P}}
\def\bmQ{{\bm Q}}

\def\notSigma{{\not\!\Sigma}}
\def\notK{{\not\!\! K}}
\def\notL{{\not\!\!L}}
\def\notD{{\not\!\!D}}
\def\notd{{\not\!d}}
\def\notl{{\not\!l}}

\def\notn{{\not\!n}}

\def\notJ{{\not\!J}}
\def\notj{{\not\!j}}

\def\fraka{\mathfrak{a}}
\def\frakb{\mathfrak{b}}
\def\frakc{\mathfrak{c}}
\def\frakd{\mathfrak{d}}
\def\frake{\mathfrak{e}}

\newcounter{mnotecount}

\newcommand{\mnotex}[1]
{\protect{\stepcounter{mnotecount}}$^{\mbox{\footnotesize $\bullet$\themnotecount}}$ 
\marginpar{
\raggedright\tiny\em
$\!\!\!\!\!\!\,\bullet$\themnotecount: #1} }

\begin{document}

\title{\textbf{Construction of anti-de Sitter-like spacetimes using the metric
conformal Einstein field equations: the tracefree matter case}}
\author{Diego A. Carranza\footnote{E-mail address:{\tt d.a.carranzaortiz@qmul.ac.uk}}\ }
\author{Juan A. Valiente Kroon\footnote{E-mail address:{\tt j.a.valiente-kroon@qmul.ac.uk}}}
\affil{School of Mathematical Sciences, Queen Mary University of London,
Mile End Road, London E1 4NS, United Kingdom.}

\maketitle

\begin{abstract}

\noindent Using a metric conformal formulation of the Einstein equations, we
develop a construction of 4-dimensional anti-de Sitter-like spacetimes coupled
to  tracefree matter models. Our strategy relies on the formulation of an
initial-boundary problem for a system of quasilinear wave equations for various
conformal fields by exploiting the conformal and coordinate gauges. By
analysing the conformal constraints we show a systematic procedure to prescribe
initial and boundary data. This analysis is complemented by the propagation of
the constraints, showing that a solution to the wave equations implies a
solution to the Einstein field equations. In addition, we study three explicit
tracefree matter models: the conformally invariant scalar field, the Maxwell
field and the Yang-Mills field.  For each one of these we identify the basic
data required to couple them to the system of wave equations. As our main
result, we establish the local existence and uniqueness of solutions for the
evolution system in a neighbourhood around the corner, provided compatibility
conditions for the initial and boundary data are imposed up to a certain order.

\end{abstract}

\section{Introduction}

The study of solutions to the Einstein equations with negative Cosmological
constant, the so-called \emph{anti-de Sitter-like} spacetimes, results
particularly challenging due to the presence of a timelike conformal boundary
located at spatial infinity. This boundary makes the spacetime non-globally
hyperbolic and, thus, it cannot be fully reconstructed from an initial value
problem. Accordingly, information on the conformal boundary must also be
provided. In this scenario, the methods of conformal geometry offer a natural
approach to face this challenge by providing an (unphysical) auxiliary
spacetime where the boundary is located at a finite distance. 

Although the use of conformal methods into General Relativity can be traced
back to Penrose's seminal work \cite{Pen63}, a satisfactory conformal
formulation of the Einstein field equations from the point of view of the
theory of partial differential equations was first obtained by Friedrich
\cite{Fri81a}. The latter enabled a systematic study of asymptotically simple
solutions to Einstein equations. In particular, it made possible a construction
of vacuum anti-de Sitter-like spacetimes where a non-metric formulation of
conformal Einstein field equations is employed \cite{Fri95, Fri14}. At the
heart of this approach is a first order symmetric hyperbolic system for a set
of conformal fields obtained by exploiting the properties of certain conformal
invariants ---the so-called conformal geodesics.  For this evolution system,
maximally dissipative boundary conditions are considered in order to establish
a result of local existence of solutions.  Despite identifying some boundary
data in a covariant manner, the type of equations used in this construction, in
conjunction with the gauge choice, makes it difficult to implement this scheme
in numerical codes. 

The latter issue becomes particularly relevant in view of the growing interest
in the conjectured instability of the anti-de Sitter spacetime ---see
\cite{Biz13} for an entry point into the subject. A considerable amount of
effort has been directed to understand the dynamics of this class of solutions
via numerical methods ---see, for example, \cite{BizRos11, BizMalRos15,
DiaHorSan12, DiaHorMarSan12, AbaSilLopMasSer14, SilLopMasSer15}.  Also,
alternative hyperbolic formulations coupling the scalar field have been
discussed in \cite{SanSop16a, SanSop16b} and a proof of this conjecture has
been obtained for the Einstein-Vlasov system with spherical symmetry
\cite{Mos18}. 

As an alternative to solving first order systems, Paetz \cite{Pae15} showed
that, in the vacuum case, it is possible to construct a second order evolution
system from the conformal Einstein equations. Based on this result, in
\cite{CarVal18b} an alternative construction for vacuum anti-de Sitter-like
spacetimes has recently been given. Using a metric version of the conformal
Einstein fields equations, a suitable choice of coordinates makes it possible
to obtain a quasilinear system of wave equations for a set of conformal fields.
Moreover, the identified free boundary data consist of a Lorentzian 3-metric on
this hypersurface together with a linear combination of the incoming and
outgoing gravitational radiation. Due to the manifestly hyperbolic nature of
the system under consideration, it is expected that this scheme can be more
easily adapted to current numerical codes.

The study of the Einstein field equations coupled to some suitable matter model
is generally carried out on a case-by-case manner. Remarkably,  in the context
of conformal methods, tracefree matter models are amenable to a more systematic
study. This class of matter models contains several cases of interest such as
the electromagnetic and Yang-Mills fields. In \cite{Fri91}, Friedrich has
provided a suitable conformal formulation of the Einstein field equations under
the assumption of a tracefree energy-momentum tensor, which, in turn, has been
exploited to establish stability results for the Einstein-Yang-Mills system
with positive and zero Cosmological constant.  Motivated by this investigation,
a local existence and uniqueness result for the same system has been presented
in \cite{LueVal14a} under the assumption of spherical symmetry. Despite the
complications from considering non-trivial matter fields, advances have been
made in a variety of scenarios using different approaches ---see, for example,
\cite{Hol12, Fri14c, Fri17}.  Nevertheless, a satisfactory conformal
formulation allowing a methodical study of general matter models remains
elusive.

In the present article we generalise the analysis in \cite{CarVal18b} and
consider the construction of anti-de Sitter-like spacetimes coupled to
tracefree matter models. Exploiting the conformal freedom we set an appropriate
gauge yielding a system of geometric wave equations for the relevant conformal
fields \cite{CarHurVal19}. A suitable coordinate choice allows us to cast
it as a system of quasilinear of wave equations, provided the matter field has
\emph{good} properties. For this type of hyperbolic evolution equations there
are available results concerning the local existence and uniqueness of
solutions ---see e.g. \cite{CheWah83,DafHru85}. The conformal constraints
supply the system with adequate initial and boundary data.  In this work three
models of tracefree matter are considered: the conformally invariant scalar
field, the Maxwell field and the Yang-Mills field. Initial and boundary data
have been identified for each one of them, as well as an analysis of the
corresponding propagation of the constraints. Given the amount of heavy
calculations some parts of this work require, the suite \texttt{xAct} of
\texttt{Mathematica} for tensorial computations has been employed ---see
\cite{xAct}.

\medskip
The main result of the article can be stated as follows:

\begin{theorem}
\label{Theorem:MainTracefree}
Let $\mathcal{S}_\star$ be a 3-dimensional spacelike hypersurface with boundary
$\partial\mathcal{S}_\star$  and  on it  smooth tracefree anti-de Sitter-like
initial data for the Einstein field equations coupled to either: (i) the
conformally invariant scalar field, (ii) the Maxwell field or (iii) the
Yang-Mills field.  Consider the cylinder $\mathscr{I}_{\tau_\bullet}\equiv[0,
\tau_\bullet) \times \partial\mathcal{S}_\star$ for some $\tau_\bullet > 0$
endowed with a set of smooth fields satisfying the conformal Einstein
constraints on $\mathscr{I}_{\tau_\bullet}$. Assume suitable basic boundary
data for the tracefree matter fields. If, in addition, the data on
$\mathcal{S}_\star$ and on $\mathscr{I}_{\tau_\bullet}$ satisfy, up to some
order, compatibility (i.e. corner) conditions at $\partial\mathcal{S}_\star$,
then there exists a smooth solution to the Einstein field equations with
$\lambda<0$ coupled to any of the aforementioned tracefree matter fields in a
neighbourhood of $\partial\mathcal{S}_\star$.
\end{theorem}

\noindent The proof of this theorem follows from the analysis of the various
sections of the article. The detailed boundary conditions for the geometric
fields are detailed in \Cref{Proposition:SummaryBoundaryData} while those for
the matter fields are detailed in Lemmas \ref{Lemma:ScalarField},
\ref{Lemma:MaxwellField} and \ref{Lemma:YM}. To the best of our knowledge, the
results in \Cref{Theorem:MainTracefree} are the first results regarding the
local existence of non-vacuum anti-de Sitter-like spacetimes in the absence of
symmetries available in the literature. 

\subsection{Outline of the article}

This article builds on the theory developed in \cite{CarVal18b} and
\cite{CarHurVal19}. In order to ease the presentation we make direct
use of the relevant results of these references and refer the reader
to them for full details. Accordingly, in this article we emphasise the novel aspects
of the analysis and how the ideas of the above references fit together.

\medskip

In \Cref{Section:TMCEFE} the basic definitions related to the conformal
formulation of the Einstein equations used throughout this work are introduced.
\Cref{Section:WE} is devoted to the discussion of the evolution system for the
conformal fields and the role the gauge choice plays in its successful
recasting as a hyperbolic system; also, the set of zero-quantities are defined.
\Cref{Section:ConformalEinsteinConstraints} presents some results about the
conformal constraint equations on spacelike and timelike hypersurfaces. In
\Cref{Section:Setup} we focus on analysing the boundary data required to
establish a well-posed problem for the system of wave equations and their
relation to the zero-quantities. In \Cref{Section:PropOfConstr} we provide a
brief discussion of how the boundary data enables us to propagate the
constraints on the conformal boundary. In \Cref{Section:MatterModels} several
explicit matter models are studied in detail, focusing on the basic data
necessary to couple them to the evolution system. Finally,
\Cref{Section:FinalRemarks} provides some final remarks.

\subsection*{Conventions}
Throughout this work, $(\tilde{\mathcal{M}},\tilde{\bmg})$ will denote a
4-dimensional Lorentzian spacetime satisfying the Einstein equations with
Cosmological constant $\lambda$.  The signature of the metric in this article
will be $(-,+,+,+)$.  Lowercase Latin letters $a,\, b,\, c, \ldots$ are used as
abstract spacetime tensor indices while the indices $i,\,j,\,k,\ldots$ are
abstract indices on the tensor bundle of hypersurfaces of
$\tilde{\mathcal{M}}$. Greek letters $\alpha,\beta,\gamma,\dots$ will be used
as coordinate indices on either spacelike or timelike hypersurfaces. Our
conventions for the curvature are
\[
\nabla_c \nabla_d u^a -\nabla_d \nabla_c u^a = R^a{}_{bcd} u^b.
\]

\section{The tracefree metric conformal Einstein field equations}
\label{Section:TMCEFE}

This section discusses some properties of the tracefree conformal Einstein
field equations, the basic tool to be used in the remainder of this article.
The system of geometric wave equations resulting from the conformal Einstein
field equations is briefly analysed. Finally, it is also discussed how a
suitable gauge choice enables us to cast these evolution equations as a
quasilinear hyperbolic system.

\subsection{Basic definitions}

Let $(\tilde{\mathcal{M}},\tilde{\bmg})$ denote a physical spacetime
satisfying the Einstein field equations
\begin{equation}
\tilde{R}_{ab} = \lambda \tilde{g}_{ab} + (\tilde{T}_{ab} - \tfrac12\tilde{g}_{ab}\tilde{T}),
\label{EFE}
\end{equation}
where $\tilde{R}_{ab}$ is the Ricci tensor of the metric $\tilde{g}_{ab}$,
$\tilde{T}_{ab}$ the energy-momentum tensor and $\tilde{T} \equiv
\tilde{g}_{ab}\tilde{T}^{ab}$ its trace. Let $(\mathcal{M},\bmg)$ be a
spacetime conformally related to $(\tilde{\mathcal{M}},\tilde{\bmg})$ via a
conformal embedding
\[
\tilde{\mathcal{M}}\stackrel{\varphi}{\hookrightarrow}\mathcal{M},
\qquad \tilde{g}_{ab} \stackrel{\varphi}{\mapsto} g_{ab}\equiv
\Xi^2 \big( \varphi^{-1})^*\tilde{g}_{ab}, \qquad
\Xi|_{\varphi(\tilde{\mathcal{M}})}>0.
\]
with $\Xi$ the so-called \emph{conformal factor}.  Abusing of the notation we
write
\begin{equation}
g_{ab} = \Xi^2 \tilde{g}_{ab}.
\label{ConformalRescaling}
\end{equation}
We will refer to $(\mathcal{M}, \bmg)$ as the \emph{unphysical spacetime}.
The set of points of $\mathcal{M}$ for which $\Xi$ vanishes define the
\emph{conformal boundary}. We use the notation $\mathscr{I}$ to denote the
parts of the conformal boundary which are a hypersurface of
$\mathcal{M}$.

In what follows let $R^a{}_{bcd}, \ R_{ab}, \ R$ and $C^a{}_{bcd}$ denote,
respectively, the Riemann tensor, the Ricci tensor, the Ricci scalar and the
(conformally invariant) Weyl tensor of the metric $g_{ab}$. For the discussion
of the metric conformal Einstein field equations we conveniently
introduce the fields
\begin{subequations}
\begin{eqnarray}
&& L_{ab}\equiv\tfrac{1}{2}R_{ab}-\tfrac{1}{12}g_{ab}R, \label{SchoutenTensorDefinition} \\
&& s\equiv \tfrac{1}{4}\nabla^{c}\nabla_{c}\Xi+\tfrac{1}{24}R\Xi, \label{FriedrichScalarDefinition}\\
&& d^a{}_{bcd} \equiv \Xi^{-1} C^a{}_{bcd}, \label{RescaledWeylTensorDefinition}
\end{eqnarray}
\end{subequations}
known, respectively, as the Schouten tensor, the Friedrich scalar and the
rescaled Weyl tensor.

\subsection{The energy-momentum tensor}

Unlike the various geometrical objects associated to the metric
$\tilde{g}_{ab}$, the rescaling \eqref{ConformalRescaling} does not determine
the transformation rule for the energy-momentum tensor. Then, for simplicity,
let define its unphysical counterpart $T_{ab}$ as
\begin{equation}
T_{ab} \equiv \Xi^{-2}\tilde{T}_{ab}.
\label{EnergyMomentumRescaling}
\end{equation}
In terms of the Levi-Civita connection of the unphysical metric $\nabla_a$,
this rescaling implies that
\[
\nabla^a T_{ab} = \Xi^{-1}T \nabla_b\Xi,
\]
where it has been used that $\tilde{T}_{ab}$ is divergencefree and $T \equiv
g_{ab}T^{ab}$. From here we have that 
\[
\nabla^a T_{ab} = 0 \quad \Longleftrightarrow \quad T = 0.
\]
\begin{remark}
{\em From equation \eqref{EnergyMomentumRescaling} it can be
checked that $T= 0$ if $\tilde{T} = 0$. In view of this it will be assumed
hereafter that the physical matter model is tracefree.}
\end{remark}

Due to the presence of a non-vanishing tracefree matter component it results
convenient to introduce the rescaled Cotton tensor
\begin{equation}
T_{abc} \equiv \Xi\nabla_{[a}T_{b]c} - g_{c[a}T_{b]d}\nabla^d\Xi - 3T_{c[a}\nabla_{b]}\Xi.
\label{CottonTensorDefinition}
\end{equation}
From this definition it follows that it possesses the symmetries $T_{abc} =
T_{[ab]c}$ and $T_{[abc]} = 0$.

\subsection{Basic properties}

\medskip
In terms of the objects defined above, the \emph{tracefree metric conformal
Einstein field equations} are given by \cite{Fri91}:
\begin{subequations}
\begin{eqnarray}
&& \nabla_a \nabla_b \Xi =-\Xi L_{ab} + sg_{ab} + \tfrac12\Xi^3T_{ab}, \label{TraceCFE1}\\
&& \nabla_a s = -L_{ac} \nabla^c \Xi + \tfrac12\Xi^2T_{ab}\nabla^b\Xi, \label{TraceCFE2}\\
&& \nabla_a L_{bc} - \nabla_b L_{ac} = \Xi T_{abc} + \nabla_e \Xi d^e{}_{cab}, \label{TraceCFE3}\\
&& \nabla_e d^e{}_{cab} = T_{abc},  \label{TraceCFE4}\\
&& \lambda  = 6 \Xi s - 3 \nabla_c \Xi \nabla^c \Xi, \label{TraceCFE5} \\
&& R^a{}_{bcd} = \Xi d^a{}_{bcd} +2 \delta_{[c}{}^a L_{d]b} + 2 L_{[c}{}^a g_{d]b}. \label{TraceCFE6}
\end{eqnarray}
\end{subequations}
A detailed derivation and discussion of this system of equations can be found in
$\cite{CFEBook}$. Additionally, as a consequence of imposing the tracefree
condition on the matter sector, the system is supplemented by the
conservation law
\begin{equation}
\nabla^aT_{ab} = 0. \label{UnphysicalConservationLaw}
\end{equation}

\begin{remark}
{\em Equations \eqref{TraceCFE1}-\eqref{TraceCFE4} are read as differential
conditions for the fields $\Xi, \ s, \ L_{ab}$ and $d^a{}_{bcd}$ subject to the
condition $\eqref{UnphysicalConservationLaw}$. Equation \eqref{TraceCFE5} plays
the role of an algebraic constraint which is satisfied if it holds at a single
point by virtue of the other equations ---see Lemma 8.1 in \cite{CFEBook}.
Equation \eqref{TraceCFE6} establishes the irreducible decomposition of the
Riemann tensor. Moreover,  it also provides an equation for
$g_{ab}$.}
\end{remark}

\begin{remark}
{\em Direct evaluation of \eqref{TraceCFE5} on the conformal boundary shows
that the causal character of $\mathscr{I}$ is completely determined by the sign
of the Cosmological constant $\lambda$. In particular, for anti de Sitter-like
spacetimes ($\lambda < 0$) $\mathscr{I}$ is timelike ---see
e.g. Theorem 10.1 in \cite{CFEBook}.}
\end{remark}

By a solution to the tracefree metric conformal Einstein field equations it is
understood a collection $(g_{ab}, \ \Xi, \ s, \ L_{ab}, \ d^a{}_{bcd}, \
T_{ab})$ satisfying equations \eqref{TraceCFE1}-\eqref{TraceCFE5} and
\eqref{UnphysicalConservationLaw}.  The relation between the metric conformal
Einstein field equations and the Einstein field equations is given by the
following statement:
\begin{proposition}
\label{Proposition:FriedrichProposition}
Let  $(g_{ab}, \ \Xi, \ s, \ L_{ab}, \ d^a{}_{bcd},\ T_{ab})$ denote a solution
to the metric tracefree conformal Einstein field equations such that $\Xi\neq
0$ on an open set $\mathcal{U}\subset \mathcal{M}$. If, in addition, equation
\eqref{TraceCFE5} is satisfied at a point $p\in \mathcal{U}$, then the metric
\[
\tilde{g}_{ab} = \Xi^{-2} g_{ab}
\]
is a solution to the Einstein field equations \eqref{EFE} on $\mathcal{U}$.
\end{proposition} A proof of this proposition is given in \cite{CFEBook} 
---see Proposition 8.1. 

\section{The evolution system}
\label{Section:WE}

In \cite{Pae15} Paetz has proved that if $T_{ab} = 0$ it is possible to obtain a
system of \emph{geometric} wave equations for the collection of conformal
fields $(\Xi,\ s,\ L_{ab}, \ d^a{}_{bcd})$. This has been generalised to the
tracefree matter case in \cite{CarHurVal19}, resulting in the following:
\begin{proposition}
\label{Proposition:CWE}
Any solution $(\Xi, \ s, \ L_{ab}, \ d^a{}_{bcd})$ to the tracefree metric
conformal Einstein field equations \eqref{TraceCFE1}-\eqref{TraceCFE4}
satisfies the equations
\begin{subequations}
\begin{eqnarray}
&& \square\Xi = 4s - \tfrac16 R, \label{CWE1} \\
&& \square s = - \tfrac{1}{6} s R + \Xi L_{ab} L^{ab} - \tfrac{1}{6} \nabla_{a}R \nabla^{a}\Xi
+ \tfrac{1}{4} \Xi^5 T_{ab} T^{ab} - \Xi^3 L_{ab}T^{ab} 
+ \Xi \nabla^{a}\Xi \nabla^{b}\Xi T_{ab}, \label{CWE2} \\
&& \square L_{ab} = -2 \Xi d_{acbd} L^{cd} + 4 L_{a}{}^{c} L_{bc} - L_{cd} L^{cd} g_{ab} 
+ \tfrac{1}{6} \nabla_{a}\nabla_{b}R + \tfrac{1}{2} \Xi^3 d_{acbd} T^{cd} \nonumber \\
&& \hspace{1.3cm} - \Xi \nabla_{c}T_{a}{}^{c}{}_{b} - 2 T_{(a|c|b)} \nabla^{c}\Xi, \label{CWE3} \\
&& \square d_{abcd} = - 4 \Xi d_{a}{}^{f}{}_{[c}{}^{e} d_{d]ebf} - 2 \Xi d_{a}{}^{f}{}_{b}{}^{e} d_{cdfe} 
+ \tfrac{1}{2} d_{abcd} R - T_{[a}{}^f \Xi^2 d_{b]fcd} - \Xi^2 T_{[c}{}^f d_{d]fab} \nonumber \\
&& \hspace{1.5cm} - \Xi^2 g_{a[c} d_{d]gbf} T^{fg} + \Xi^2 g_{b[c}  d_{d]gaf} T^{fg} 
+ 2 \nabla_{[a}T_{|cd|b]} + \epsilon_{abef} \nabla^{f}\,^*T_{cd}{}^{e}. \label{CWE4}
\end{eqnarray}
\end{subequations}
\end{proposition} 

\begin{remark}
{\em The system \eqref{CWE1}-\eqref{CWE4} needs to be supplemented with a wave
equation for the unphysical metric. Taking the trace of \eqref{TraceCFE6} 
we obtain
\begin{equation}
R_{ab} = 2 L_{ab} + \frac{1}{6}Rg_{ab}.
\label{EquationMetric}
\end{equation}
Here, $R_{ab}$ encodes the derivatives of $g_{ab}$ once a system of coordinates
has been adopted, while $L_{ab}$ is an independent field satisfying the system
\eqref{CWE1}-\eqref{CWE4}. The last equation can be viewed as an Einstein field
equation for $g_{ab}$ coupled to some unphysical matter fields. In view of
this, equation \eqref{EquationMetric} be known as the \emph{unphysical Einstein
equation}.}
\end{remark}

\begin{remark}
\label{ProblemsCWE}
{\em Equations \eqref{CWE1}-\eqref{CWE4} and \eqref{EquationMetric} are not
completely satisfactory wave equations for a number of reasons: (i) the
appearance of second order derivatives of $R$; (ii) once coordinates are
introduced, undesired second order derivatives of the components of $g_{ab}$
appear in the principal part of equations \eqref{CWE3}, \eqref{CWE4} and
\eqref{EquationMetric}; (iii) when a particular matter model is considered,
second (and higher) order derivatives of the corresponding field will emerge in
the coupled system.  In order to apply results from the theory of partial
differential equations, these problems must be addressed. Issues (i) and (ii)
have been discussed in \cite{CarVal18b} and will be briefly presented in the
next subsection, while (iii) has been analysed in \cite{CarHurVal19} for some
particular models.}
\end{remark}

\subsection{Gauge considerations}
\label{Section:Gauge}

The conformal Einstein field equations possess both a coordinate and a
conformal freedom which can be exploited to recast the geometric wave equations
\eqref{CWE1}-\eqref{CWE4} as satisfactory hyperbolic evolution equations. This
subsection is based on \cite{CarVal18b}, Section 2.3, where a more detailed
presentation and the motivation behind the relevant definitions can be
found.

\subsubsection{Conformal gauge source function}

Given two conformally related metrics $g_{ab}$ and $g'_{ab}$, the
relation between their corresponding Ricci scalars can be read as a
differential equation for their conformal factor $\theta$. Then,
prescribing the Ricci scalar is equivalent to the specification of the
representative of the conformal class $[\tilde{\bmg}]$.  Accordingly,
the Ricci scalar associated to the unphysical metric $g_{ab}$ will be
specified by the gauge function $\mathcal{R}(x)$ ---that is
\[
R = \mathcal{R}(x).
\]

\subsubsection{Generalised harmonic coordinates and the reduced Ricci operator}

The other piece of the gauge freedom is contained in the coordinate choice.
Let $x = (x^\mu)$ be coordinates satisfying the \emph{generalised wave}
condition
\begin{equation}
\Box x^\mu = - \mathcal{F}^\mu(x),
\label{GeneralisedWaveCoordinateCondition}
\end{equation}
where $\mathcal{F}^\mu(x)$ are the so-called \emph{coordinate gauge source
functions}. A simple
calculation shows that these are related to the \emph{contracted Christoffel
symbols} $\Gamma^\mu \equiv g^{\sigma\tau}\Gamma^\mu{}_{\sigma\tau}$ via
\[
\mathcal{F}^\mu(x) = -\Gamma^\mu.
\]

In this context, we introduced the \emph{reduced Ricci operator}
\begin{equation}
\mathscr{R}_{\mu\nu}[\bmg] \equiv R_{\mu\nu} - g_{\sigma(\mu}\nabla_{\nu)}\Gamma^\sigma +
g_{\sigma(\mu}\nabla_{\nu)} \mathcal{F}^\sigma(x).
\end{equation}
Here, $R_{\mu\nu}$ is given explicitly in terms of first and second order
partial derivatives of $g_{\mu\nu}$.  Thus, by choosing coordinates satisfying
condition \eqref{GeneralisedWaveCoordinateCondition}, the unphysical Einstein
equations \eqref{EquationMetric} takes the form
\begin{equation}
\mathscr{R}_{\mu\nu}[\bmg] = 2 L_{\mu\nu} +\frac{1}{6}\mathcal{R}(x) g_{\mu\nu}.
\label{UnphysicalEinsteinEquation}
\end{equation}
By choosing generalised wave coordinates, this equation does correspond to a
quasilinear wave equation for the components of $g_{ab}$.

\subsubsection{The reduced wave operator}

As remarked before, when the operator $\square$ for an arbitrary coordinates
choice is applied on a tensorial field, undesired second derivatives of
$g_{\mu\nu}$ appear. In this sense this operator is not suitable if one wants
to preserve the hyperbolicity of the system. Making use of the generalised wave
coordinate condition, a suitable operator can be defined as follows:

\begin{definition}
Let $T_{\lambda\cdots\rho}$ be a tensor field. The action of the \emph{reduced
wave operator}, $\blacksquare$, on it is given by:
\begin{eqnarray*}
&& \blacksquare T_{\lambda \cdots \rho} \equiv \square T_{\lambda\cdots\rho} +\bigg( (2 \Phi_{\tau\lambda} 
+ \frac{1}{4}\mathcal{R}(x) g_{\tau\lambda} - R_{\tau\lambda}) 
- g_{\sigma\tau}\nabla_\lambda (\mathcal{F}^\sigma(x)-\Gamma^\sigma) \bigg)T^\tau{}_{\cdots\rho} + \cdots\\
&& \hspace{3cm} \cdots + \bigg( (2 \Phi_{\tau\rho} + \frac{1}{4}\mathcal{R}(x)
g_{\tau\rho} - R_{\tau\rho}) - g_{\sigma\tau}\nabla_\rho (\mathcal{F}^\sigma(x)-\Gamma^\sigma)
\bigg)T_{\lambda\cdots}{}^\tau
\end{eqnarray*}
where $\square \equiv g^{\mu\nu}\nabla_\mu\nabla_\nu$. The action of
$\blacksquare$ on a scalar field $\phi$ is simply given by
\[
\blacksquare \phi \equiv g^{\mu\nu}\nabla_\mu\nabla_\nu \phi. 
\]
\end{definition}
\noindent As shown in \cite{CarVal18b}, if generalised wave coordinates are chosen then
$\blacksquare$ represents a suitable second order hyperbolic operator for the
system \eqref{CWE1}-\eqref{CWE4}. In this sense, we will say that a system of
wave equations expressed in terms of $\blacksquare$ is \emph{proper}.

\subsubsection{Summary: gauge reduced evolution equations}

In view of the previous discussion, and as presented in \cite{CarHurVal19},
when generalised wave coordinates are adopted, the system of wave equations for
the components of the conformal geometric fields  to be considered is the following:
\begin{subequations}
\begin{eqnarray}
&&\hspace{-13mm} \blacksquare \Xi = 4s -\frac{1}{6}\Xi \mathcal{R}(x), \label{ReducedWaveCFE1} \\
&&\hspace{-13mm} \blacksquare s = - \tfrac{1}{6} s \mathcal{R}(x) + \Xi L_{\mu\nu} L^{\mu\nu} 
- \tfrac{1}{6} \nabla_{\mu}\mathcal{R}(x)
\nabla^{\mu}\Xi + \tfrac{1}{4} \Xi^5 T_{\mu\nu}T^{\mu\nu} - \Xi^3 L_{\mu\nu}T^{\mu\nu} 
+ \Xi \nabla^{\mu}\Xi \nabla^{\nu}\Xi T_{\mu\nu},
\label{ReducedWaveCFE2} \\
&&\hspace{-13mm} \blacksquare L_{\mu\nu} = -2 \Xi d_{\mu\rho\nu\lambda} L^{\rho\lambda} 
+ 4 L_{\mu}{}^{\lambda} L_{\nu\lambda} - L_{\lambda\rho} L^{\lambda\rho} g_{\mu\nu} 
+ \tfrac{1}{6} \nabla_{\mu}\nabla_{\nu}\mathcal{R}(x) + \tfrac{1}{2} \Xi^3 d_{\mu\lambda\nu\rho} T^{\lambda\rho} 
\nonumber \\
&& \hspace{0.1cm} - \Xi \nabla_{\lambda}T_{\mu}{}^{\lambda}{}_{\nu} 
- 2 T_{(\mu|\lambda|\nu)} \nabla^{\lambda}\Xi,\label{ReducedWaveCFE3} \\
&&\hspace{-13mm} \blacksquare d_{\mu\nu\lambda\rho} = - 4 \Xi
d_{\mu}{}^{\tau}{}_{[\lambda}{}^{\sigma} d_{\rho]\sigma\nu\tau} - 2 \Xi
d_{\mu}{}^{\tau}{}_{\nu}{}^{\sigma} d_{\lambda\rho\tau\sigma} + \tfrac{1}{2}
d_{\mu\nu\lambda\rho} \mathcal{R}(x) - T_{[\mu}{}^\sigma \Xi^2 d_{\nu]\sigma\lambda\rho} -
\Xi^2 T_{[\lambda}{}^\sigma d_{\rho]\sigma\mu\nu} \nonumber \\
&& \hspace{0.3cm} - \Xi^2 g_{\mu[\lambda} d_{\rho]\sigma\nu\tau} T^{\tau\sigma} 
+ \Xi^2 g_{\nu[\lambda}  d_{\rho]\sigma\mu\tau} T^{\tau\sigma} 
+ 2 \nabla_{[\mu}T_{|\lambda\rho|\nu]} + \epsilon_{\mu\nu\sigma\tau} \nabla^{\tau}\,^*T_{\lambda\rho}{}^{\sigma},
\label{ReducedWaveCFE4}\\
&& \hspace{-13mm}\mathscr{R}_{\mu\nu}[\bmg] = 2 L_{\mu\nu}
   +\frac{1}{6}\mathcal{R}(x) g_{\mu\nu}. \label{ReducedWaveCFE5}
\end{eqnarray}
\end{subequations}

\begin{remark}
{\em Since the operator $\blacksquare$ enables us to remove all the undesired
second order partial derivatives of $g_{\mu\nu}$, the reduced system
\eqref{ReducedWaveCFE1}-\eqref{ReducedWaveCFE5} constitutes a system of
quasilinear wave equations for the fields $\Xi$, $s$, $L_{\mu\nu}$,
$d_{\mu\nu\lambda\rho}$ and $g_{\mu\nu}$. Strictly speaking though, this is a
system of wave equations only if $g_{\mu\nu}$ is known to be Lorentzian.
Precise statements concerning the local existence and uniqueness results for
this type of systems can be found in \cite{CheWah83, DafHru85}.}
\end{remark}

\begin{remark}
{\em One has to prove that the gauge introduced in this section is consistent
with the field equations. In \cite{CarHurVal19} it has been shown how to
construct a system of homogeneous wave equations for an additional set of
subsidiary fields under the assumption of a tracefree energy-momentum tensor.
The conditions required to establish vanishing initial and boundary data for
these fields have been described in \cite{CarVal18b}, from where the
propagation of the gauge follows.}
\end{remark}

\subsection{The subsidiary evolution equations}

Once a suitable system of wave equations for the conformal fields has been
obtained, it is necessary to prove that any solution to this system corresponds
to a solution to the metric tracefree conformal Einstein field equations
\eqref{TraceCFE1}-\eqref{TraceCFE4}. This can be achieved by first defining the
set of \emph{geometric zero-quantities}
\begin{subequations}
\begin{eqnarray}
&& \Upsilon_{ab} \equiv \nabla_a \nabla_b \Xi + \Xi L_{ab} - sg_{ab} - \frac13\Xi^3 T_{ab}, \label{ZQ1}\\
&& \Theta_a \equiv \nabla_a s + L_{ac} \nabla^c \Xi - \frac12\Xi^2 T_{ab} \nabla^b\Xi, \label{ZQ2}\\
&& \Delta_{abc} \equiv \nabla_a L_{bc} - \nabla_b L_{ac} - \nabla_d \Xi d^d{}_{cab}
- \Xi T_{abc}, \label{ZQ3}\\
&& \Lambda_{abc}\equiv T_{bca} - \nabla_e d^e{}_{abc}, \label{ZQ4} \\
&& Z \equiv \lambda - 6\Xi s +3 \nabla^c\Xi\nabla_c\Xi, \label{ZQ5} \\
&& P^a{}_{bcd} \equiv R^a{}_{bcd} - \Xi d^a{}_{bcd} - 2 \delta_{[c}{}^a L_{d]b} - 2 L_{[c}{}^a g_{d]b}. \label{ZQ6}
\end{eqnarray}
\end{subequations}
In terms of these fields, the metric tracefree conformal Einstein field
equations can be expressed as
\begin{equation}
\Upsilon_{ab}=0, \qquad \Theta_a =0, \qquad \Delta_{abc}=0, \qquad
\Lambda_{abc}=0 \qquad Z = 0, \qquad P^a{}_{bcd} = 0.
\label{CFEZeroQuantities}
\end{equation}
In \cite{CarHurVal19} it has been shown that the evolution of these quantities
can be encoded in a system of homogeneous geometric wave equations. This result
can be enunciated as follows:

\begin{lemma}
\label{Lemma:WEZQ}
Assume that the conformal fields $\Xi$, $s$, $L_{ab}$ and $d_{abcd}$ satisfy
the geometric wave equations \eqref{CWE1}-\eqref{CWE4}. Then the geometric
zero-quantities satisfy a system of geometric wave equations of the form
\begin{subequations}
\begin{eqnarray}
&& \square \Upsilon_{ab} = H_{ab}({\bm\Upsilon}, {\bm\nabla\bm\Upsilon},{\bmP}),
\label{SubsidiaryEquation1}\\
&& \square \Theta_a = H_a({\bm\Upsilon}, {\bm\Theta}, {\bm\Delta}, {\bmP}), \label{SubsidiaryEquation2}\\
&& \square \Delta_{abc} = H_{abc}(\bm\Upsilon, \bm\nabla\bm\Upsilon, \bm\Theta, \bm\Delta,
 \bm\Lambda, \bm\nabla\bm\Lambda, \bmP, \bm\nabla\bmP),
 \label{SubsidiaryEquation3}\\
&& \square \Lambda_{abc} = L_{abc}(\bm\Upsilon, \bm\nabla\bm\Upsilon, \bm\Delta,
\bm\Lambda, \bmP, \bm\nabla\bmP), \label{SubsidiaryEquation4} \\
&& \square Z = H(\bm\Upsilon, \bm\Theta), \label{SubsidiaryEquation5} \\
&& \square P_{abcd} = H_{abcd}(\bm\Upsilon, \bm\nabla\bm\Delta, \bm\Lambda, \bm\nabla\bm\Lambda,
\bmP). \label{SubsidiaryEquation6}
\end{eqnarray}
\end{subequations}
where $H_a$, $H_{ab}$, $H_{abc}$, $L_{abc}$, $H$ and $H_{abcd}$ are homogeneous
expressions of their arguments.
\end{lemma}

\begin{remark}
{\em In order to cast equations
\eqref{SubsidiaryEquation1}-\eqref{SubsidiaryEquation6} as a proper quasilinear
system of wave equations, the operator $\square$ must be replaced by the
reduced wave operator $\blacksquare$ once generalised wave coordinates have
been chosen.}
\end{remark}

\section{The conformal constraint equations}
\label{Section:ConformalEinsteinConstraints}

This section will be devoted to the discussion of the constraint equations
implied by the metric tracefree conformal Einstein field equations. This
system, first discussed in \cite{Fri84}, will serve to construct the initial
and boundary data sets required to establish a well-posed initial-boundary
problem for the wave equations \eqref{ReducedWaveCFE1}-\eqref{ReducedWaveCFE5}.
A detailed discussion of their derivation and properties in the general matter
case can be found in \cite{CFEBook}, Chapter 11. 

\subsection{Basic definitions}

Let $\mathcal{H}$ denote a (spacelike or timelike) hypersurface of the
unphysical spacetime $(\mathcal{M},\bmg)$ with unit normal vector $n_a$.
Then, we define the norm of $n_a$ as:
\[
\epsilon\equiv n_a n^a,
\]
so that $\epsilon$ take the values 1 or -1 for timelike or spacelike
hypersurfaces $\mathcal{H}$, respectively. The normal vector induces a
decomposition via the projector to $\mathcal{H}$
\[
h_a{}^b \equiv g_a{}^b- \epsilon n_a n^b.
\]
The intrinsic derivative $D_a$ on $\mathcal{H}$ is defined in the
following way: let $f$ be a scalar function and $A_a{}^b$ be a tensor field on
$\mathcal{M}$, then
\begin{eqnarray}
&& D_a f \equiv h_a{}^b \nabla_b f, \nonumber \\
&& D_e A_a{}^b  \equiv h_e{}^f h_a{}^c h_d{}^b \nabla_f A_c{}^d. \nonumber
\end{eqnarray}
Expressions involving higher rank tensors follow an analogous
rule. The derivative in the direction of $n^a$ (simply called the
normal derivative) is given by
\[
D \equiv n^a\nabla_a.
\]
The extrinsic curvature of $\mathcal{H}$ is defined as
\[
K_{ab} \equiv h_a{}^c h_b{}^d \nabla_c n_d.
\]

\medskip

In the following let 
\[
\Sigma, \quad s, \quad h_{ij} \quad L_i, \quad L_{ij}, \quad d_{ij},
\quad d_{ijk}, \quad d_{ijkl} 
\]
denote, respectively, the pull-backs of the following geometric objects
\begin{eqnarray*}
& n^a \nabla_a \Xi, \quad s, \quad g_{ab}, \quad n^c h_a{}^d L_{cd},
  \quad h_a{}^ch_b{}^d L_{cd},& \\
& n^b n^d h_e{}^a h_f{}^cd_{abcd}, \quad n^b h_e{}^a h_f{}^c
h_g{}^d d_{abcd}, \quad h_e{}^ah_f{}^bh_g{}^ch_h{}^d d_{abcd} & 
\end{eqnarray*}
to $\mathcal{H}$. In order to take into account the
contributions from the matter fields, let $\rho, \ j_i, \ T_{ij}, \ J_i, \
J_{ij}$ and $T_{ijk}$ stand, respectively, for the pull-backs of the
projections
\[
n^an^bT_{ab}, \quad n^b h_c{}^a T_{ab},\quad h_c{}^a h_d{}^b T_{ab}, 
\quad n^b n^c h_d{}^a T_{abc}, \quad n^c h_d{}^a h_e{}^b T_{abc}, \quad h_d{}^a h_e{}^b h_f{}^c T_{abc}. 
\]

\begin{remark}
{\em The tensor $h_{ij}$ corresponds to the 3-metric induced by $g_{ab}$ on $\mathcal{H}$
and it will be either Lorentzian if $\mathcal{H}$ is timelike or Riemannian if
it is spacelike.  Similarly, we will denote by $K_{ij}$ the pull-back of
$K_{ab}$ and by $K=h^{ij}K_{ij}$ its trace.}
\end{remark}

\begin{remark}
{\em The fields $d_{ij}$ and $d_{ijk}$ encode, respectively, the
\emph{electric} and \emph{magnetic parts} of the rescaled Weyl tensor
$d_{abcd}$ with respect to the normal $n_a$. It can be verified that
\begin{eqnarray*}
& d_{ij}=d_{ji}, \quad d_i{}^i=0, \quad d_{ijk}=-d_{ikj}, \quad d_{[ijk]}=0, &\\
& d_{ijkl} = 2 \epsilon(h_{i[l} d_{k]j} + h_{j[k}d_{l]i}).
\end{eqnarray*}}
\end{remark}

From this point onwards, the restriction of the conformal factor $\Xi$ to the
conformal boundary will be denoted by $\Omega$ ---that is,
\[
\Omega \equiv \Xi|_{\mathscr{I}}.
\]
Next, we define the Schouten tensor of the intrinsic 3-metric
$h_{ij}$, 
\[
l_{ij} \equiv r_{ij} -\frac{1}{4}r h_{ij},
\]
where $r_{ij}$ and $r$ are the corresponding intrinsic Ricci tensor and scalar.
In terms of these fields, a number of constraints are obtained from the
different projections of equations \eqref{TraceCFE1}-\eqref{TraceCFE5} to
$\mathcal{H}$. This results in the intrinsic equations 
\begin{subequations}
\begin{eqnarray}
&& D_i D_j \Omega = -\epsilon \Sigma K_{ij} -\Omega L_{ij} + s
h_{ij} + \frac12\Omega^3 T_{ij}, \label{ConformalConstraint1} \\
&& D_i \Sigma = K_i{}^k D_k \Omega -\Omega L_i + \frac12\Omega^3 j_i, \label{ConformalConstraint2} \\
&& D_i s = -\epsilon L_i \Sigma - L_{ik} D^k \Omega +
\frac12\Omega^2(\epsilon\Sigma j_i + T_{ij}D^j\Omega), \label{ConformalConstraint3}\\
&& D_i L_{jk} -D_j L_{ik} = -\epsilon \Sigma d_{kij}
+ D^l \Omega d_{lkij} -\epsilon (K_{ik} L_j
- K_{jk} L_i) + \Omega T_{ijk}, \label{ConformalConstraint4}\\
&& D_i L_j - D_j L_i = D^l \Omega d_{lij} +
K_i{}^k L_{jk} - K_j{}^k L_{ik} + \Omega J_{ij}, \label{ConformalConstraint5}\\
&& D^k d_{kij} =\epsilon \big(K^k{}_i d_{jk}
   -K^k{}_j d_{ik}\big) + J_{ij}, \label{ConformalConstraint6}\\
&& D^i d_{ij}= K^{ik} d_{ijk} + J_{i}, \label{ConformalConstraint7}\\
&& \lambda = 6 \Omega s - 3\epsilon \Sigma^2 - 3 D_k \Omega D^k\Omega. \label{ConformalConstraint8}
\end{eqnarray}
\end{subequations}
This system is supplemented by two geometric constraints arising from
\eqref{TraceCFE6}, namely, the conformal Codazzi-Mainardi and Gauss-Codazzi
equations:
\begin{subequations}
\begin{eqnarray}
&& D_j K_{ki} - D_k K_{ji} = \Omega d_{ijk} + h_{ij} L_k -
 h_{ik}L_j, \label{ConformalConstraint9}\\  
&& l_{ij} = -\epsilon\Omega d_{ij} + L_{ij} + \epsilon \bigg( K \big(
 K_{ij} -\displaystyle\frac{1}{4} K h_{ij}\big) - K_{ki}
  K_j{}^k + \displaystyle\frac{1}{4}  K_{kl} K^{kl}h_{\bmi\bmj}\bigg). \label{ConformalConstraint10}
\end{eqnarray}
\end{subequations}
Relations \eqref{ConformalConstraint1}-\eqref{ConformalConstraint8},
\eqref{ConformalConstraint9} and \eqref{ConformalConstraint10} are called the
\emph{metric tracefree conformal constraint equations}. A full derivation of this
system can be found in \cite{CFEBook} along with an extensive discussion of
their properties.

\subsubsection{The conformal constraints on the conformal boundary}
\label{Section:ConformalConstraintsConformalBoundary}

Let $\ell_{ab}$ be the intrinsic Lorentzian 3-metric associated to the
conformal boundary with normal vector $\notn^a$. In the following, objects
and operators crossed by a line $\slash$ will represent projections obtained
via $\ell_a{}^b$ and $\notn^a$.  Also, $\simeq$ will denote an equality valid
on $\mathscr{I}$. Evaluating the conformal constraints on the conformal
boundary ($\Omega = 0, \ \epsilon = 1$) we obtain the following simplified
system:
\begin{subequations}
\begin{eqnarray}
&& \not\!\Sigma \notK_{ij} \simeq s \ell_{ij}, \label{BdyConformalConstraint1} \\
&& \notD_i \notSigma \simeq 0, \label{BdyConformalConstraint2} \\
&& \notD_i s \simeq -\notL_i \notSigma, \label{BdyConformalConstraint3}\\
&& \notD_i \notL_{jk} - \notD_j \notL_{ik} \simeq -\notSigma \notd_{kij}
 - (\notK_{ik} \notL_j - \notK_{jk} \notL_i), \label{BdyConformalConstraint4}\\
&& \notD_i \notL_j - \notD_j \notL_i \simeq \notK_i{}^k \notL_{jk} - \notK_j{}^k \notL_{ik} , \label{BdyConformalConstraint5}\\
&& \notD^k \notd_{kij} \simeq \notK^k{}_i \notd_{jk} - \notK^k{}_j \notd_{ik} + \notJ_{ij}, \label{BdyConformalConstraint6}\\
&& \notD^j \notd_{ij} \simeq \notK^{jk} \notd_{jik} + \notJ_{i}, \label{BdyConformalConstraint7}\\
&& \lambda \simeq - 3\notSigma^2, \label{BdyConformalConstraint8}\\
&& \notD_j \notK_{ki} - \notD_k \notK_{ji} \simeq \ell_{ij} \notL_k - \ell_{ik}\notL_j, \label{BdyConformalConstraint9}\\
&& \notl_{ij} \simeq \notL_{ij} + \notK \big(\notK_{ij} -\displaystyle\frac{1}{4} \notK \ell_{ij}\big) - \notK_{ki}
  \notK_j{}^k + \displaystyle\frac{1}{4}  \notK_{kl} \notK^{kl}\ell_{ij}. \label{BdyConformalConstraint10}
\end{eqnarray}
\end{subequations}
A procedure to solve these equations in the vacuum case has been discussed in
\cite{Fri95} where the solution is given in terms of a gauge quantity related
to the Friedrich scalar and the rescaled Cotton tensor of $\ell_{ab}$. The
latter is defined as
\[
y_{ijk} \equiv \notD_i \notl_{jk} - \notD_j \notl_{ik}.
\]
Adopting this approach, the following result can be enunciated:

\begin{proposition}
\label{Proposition:ConformalConstraintsConformalBoundary}
Let $(\mathcal{M}, \bmg)$ be a 4-dimensional manifold and $\mathscr{I} \subset
\mathcal{M}$ a timelike conformal boundary with intrinsic 3-dimensional
Lorentzian metric $\ell_{ij}$ and normal $\notn^a$. Consider a tracefree
energy-momentum tensor $T_{ab}$ with $\notj_i$ its orthogonal-normal
projection with respect to $\notn^a$. Let $\varkappa(x)$ be a smooth scalar
gauge function defined on $\mathscr{I}$. Then a solution to the tracefree
conformal constraint equations
\eqref{BdyConformalConstraint1}-\eqref{BdyConformalConstraint10} on
$\mathscr{I}$ is given by the fields
\begin{subequations}
\begin{eqnarray}
&& \notSigma \simeq \sqrt{\frac{|\lambda|}{3}}, \label{SolutionLambda}\\
&& s \simeq \notSigma \varkappa, \label{SolutionConstraints0}\\
&& \notK_{ij} \simeq \varkappa \ell_{ij}, \label{SolutionConstraints1}\\
&& \notL_i \simeq - \notD_i \varkappa, \label{SolutionConstraints2}\\
&& \notL_{ij} \simeq \notl_{ij} - \frac{1}{2}\varkappa^2
   \ell_{ij}, \label{SolutionConstraints3}\\
&& \notd_{kij} \simeq - \notSigma^{-1} y_{ijk}, \label{SolutionConstraints4}
\end{eqnarray}
\end{subequations}
along with a tracefree symmetric tensor field $\notd_{ij}$ satisfying
\begin{equation}
\notD^j\notd_{ij} \simeq -\notSigma \notj_i \label{DiffEqWeylMatter}.
\end{equation}
\end{proposition}

\begin{proof}
Firstly, $\notSigma$ is given by \eqref{BdyConformalConstraint8}.  As mentioned
before, $s$ is a gauge quantity on $\mathscr{I}$ and expressed by equation
$\eqref{SolutionConstraints0}$. Direct substitution into constraints
\eqref{BdyConformalConstraint1}, \eqref{BdyConformalConstraint3},
\eqref{BdyConformalConstraint10} and \eqref{BdyConformalConstraint4} readily
leads to the solutions for $\notK_{ij}, \ \notL_i, \ \notL_{ij}$ and
$\notd_{ijk}$, respectively. Using these, equations
\eqref{BdyConformalConstraint2}, \eqref{BdyConformalConstraint5} and
\eqref{BdyConformalConstraint9} are trivially satisfied.  Regarding the
equations with matter terms, when the fields $\notJ_{ij}$ and $\notJ_i$ are
written explicitly in terms of the energy-momentum tensor via equation
\eqref{CottonTensorDefinition}, a straightforward calculation yields
\[
\notJ_{ij} \simeq 0, \quad \notJ_i \simeq - \notSigma \notj_i. 
\]
On the other hand, by virtue of the definition of $y_{ijk}$, it follows that
that $\notD^ky_{ijk} \simeq 0$. Using this and the expressions for $\notJ_i$
and $\notJ_{ij}$ stated above, it is found that \eqref{BdyConformalConstraint6}
is trivially satisfied and \eqref{BdyConformalConstraint7} corresponds to
equation \eqref{DiffEqWeylMatter}.
\end{proof}

Having obtained the solutions for the constraint equations on the conformal
boundary, a converse-like result can be formulated with the addition of an
auxiliary assumption:

\begin{lemma}
\label{Lemma:ConverseProposition}
Let $\mathcal{T} \subset \mathcal{M}$ be a timelike hypersurface such that
conditions \eqref{SolutionLambda}-\eqref{SolutionConstraints3} hold. If
$\Omega=0$ on some fiduciary spacelike hypersurface $\mathcal{C}_\star$ of
$\mathcal{T}$, then one has that $\Omega =0$ on $\mathcal{T}.$
\end{lemma}

\begin{proof}
Consider the case when $\varkappa\neq 0$ on $\mathcal{C}_\star$.  Using
equations \eqref{SolutionConstraints0}, \eqref{SolutionConstraints1} and
\eqref{SolutionConstraints3}, the trace of the conformal constraint
\eqref{ConformalConstraint1} provides with the following wave equation for
$\Omega$ to be satisfied on $\mathcal{T}$:
\begin{equation}
\notD_i \notD^i \Omega \equiv \square_\ell \Omega = 
-\Omega \bigg(\frac{r}{4} - \frac32 \varkappa^2\bigg) - \frac12\Omega^3\rho,
\label{BoxOmega}
\end{equation}
where it has been used that for a tracefree unphysical energy-momentum tensor
$T_i{}^i = -\epsilon\rho$. On the other hand, when \eqref{SolutionLambda},
\eqref{SolutionConstraints1} and $\eqref{SolutionConstraints2}$ are substituted
into constraint $\eqref{ConformalConstraint2}$ we have
\begin{equation}
\varkappa \notD_i \Omega = -\Omega \notD_i\varkappa - \frac12\Omega^3 j_i.
\end{equation}
As $\varkappa\neq 0$ and $\Omega=0$ on $\mathcal{C}_\star$, it follows from the
last equation that $\notD_i\Omega =0$ on $\mathcal{C}_{\star}$, which
represents a first order initial condition for $\Omega$.  Due to the
homogeneity of \eqref{BoxOmega}, along with the uniqueness of its solutions, we
conclude then that $\Omega =0$ on $\mathcal{T}$, that is to say, it corresponds
to the conformal boundary.  Regarding the case $\varkappa = 0$, it has been
showed in \cite{CarVal18b} how the conformal gauge freedom can be exploited to
render this case into the $\varkappa \neq 0$ one.
\end{proof}

\subsubsection{Solutions to the conformal constraints on a spacelike hypersurface}

Constraint equations
\eqref{ConformalConstraint1}-\eqref{ConformalConstraint10} enable us to obtain
the conformal version of the so-called Hamiltonian and Momentum constraints on
a spacelike hypersurface $(\epsilon = -1)$. A straightforward calculation shows
that for a tracefree matter field these take the form:
\begin{subequations}
\begin{eqnarray}
&& \frac{\Omega}{2}^2 (r + K^2 - K_{ij} K^{ij}) = 2 K \Omega \Sigma - 2 \Omega D_{i}D^{i}\Omega 
- 3 \Sigma^2 + 3 D_{i}\Omega D^{i}\Omega  + \lambda + \Omega^4\rho, \\
&& \Omega (D_{j}K_{i}{}^{j} - D_{i}K) = 2 (K_{ij} D^{j}\Omega - D_{i}\Sigma) + \Omega^3 j_{i}.
\end{eqnarray}
\end{subequations}
It follows that under a conformal approach, the collection of fields $(h_{ij},
K_{ij}, \Omega, \Sigma, \rho, j_i)$ satisfying the previous equations must be
prescribed as the basic initial data, i.e. they determine the remaining fields
on a spacelike hypersurface. Together with the boundary data, this set will
serve to evolve the wave equations for the conformal fields. Directly from the
conformal constraints one obtains the following expressions for the initial
data:
\begin{subequations}
\begin{eqnarray}
&& s =\frac{1}{3} \bigg( \Delta \Omega + \frac{1}{4}\Omega\big(r +K^2
   - K_{ij}K^{ij} \big)-\Sigma K  + \frac12 \Omega^3\rho \bigg), \label{InitialData1}\\
&& L_{ij} = \frac{1}{\Omega}\bigg( sh_{ij} + \Sigma K_{ij} -D_i D_j \Omega  \bigg) +\frac12\Omega^2T_{ij}, \label{InitialData2}\\
&& L_i = \frac{1}{\Omega}\big( K_i{}^k D_k\Omega - D_i \Sigma  \big) + \frac12\Omega^2 j_i, \label{InitialData3}
  \\
&& d_{ij} = \frac{1}{\Omega}\bigg( -L_{ij} +l_{ij} +\big( K \big(
 K_{ij} -\displaystyle\frac{1}{4} K h_{ij}\big) - K_{ki}
  K_j{}^k + \displaystyle\frac{1}{4}  K_{kl} K^{kl}h_{\bmi\bmj}\big)
   \bigg), \label{InitialData4}\\
&& d_{ijk} = \frac{1}{\Omega}\big( D_j K_{ki}- D_k K_{ji} + h_{ik}L_j
   - h_{ij} L_k  \big). \label{InitialData5}
\end{eqnarray}
\end{subequations}
The fact that these expressions are singular at $\Omega = 0$ leads to the
following:

\begin{definition}[\textbf{anti-de Sitter-like initial data with
    tracefree matter}]
\label{Definition:AdSData}
For a tracefree anti-de Sitter-like initial data set it is understood a
3-manifold $\mathcal{S}_\star$ with boundary $\partial \mathcal{S}_\star\approx
\mathbb{S}^2$ together with a collection of smooth fields
$(h_{ij},K_{ij},\Omega,\Sigma, \rho, j_i)$ such that:
\begin{itemize}
\item[(i)] $\Omega>0$ on $\mbox{\em int}\, \mathcal{S}_\star$;
\item[(ii)] $\Omega=0$ and $|\mbox{\em d}\Omega|^2=\Sigma^2-\tfrac{1}{3}\lambda>0$ on
  $\partial\mathcal{S}_\star$; 
\item[(iii)] the fields $s$, $L_{ij}$, $L_i$, $d_{ij}$ and $d_{ijk}$
  computed from relations \eqref{InitialData1}-\eqref{InitialData5} extend smoothly to
  $\partial \mathcal{S}$.
\end{itemize}
\end{definition}

\begin{remark}
{\em Anti-de Sitter-like initial data sets are closely related to so-called
hyperboloidal data sets for Minkowski-like spacetimes ---see \cite{Kan96a}. By
means of this correspondence it is possible to adapt the existence results for
hyperboloidal initial data sets in \cite{AndChrFri92,AndChr94} to the anti-de
Sitter-like setting. In particular, this shows the existence of a large class
of time symmetric data, i.e. data for which $K_{ij}=0$.}
\end{remark}

\section{General set-up}
\label{Section:Setup}

In this section we identify the relevant boundary data required for the
formulation of the initial-boundary problem for anti-de Sitter-like spacetimes
with tracefree matter. Except for the data corresponding to the electric part
of the Weyl tensor, the rest of the construction is identical to the one for
the vacuum case analysed in \cite{CarVal18b} where a more detailed discussion
is presented.

\medskip Let $(\mathcal{M},\bmg,\Xi)$ be a conformal extension of an anti-de
Sitter-like spacetime $(\tilde{\mathcal{M}},\tilde{\bmg})$ where $g_{ab}$ and
$\tilde{g}_{ab}$ are conformally related metrics.  Let $\mathcal{S}_\star \subset
\mathcal{M}$ be a smooth, compact and oriented spacelike hypersurface with
boundary $\partial\mathcal{S}_\star$. Furthermore, $\mathcal{S}_\star \cap \mathscr{I} =
\partial\mathcal{S_\star}$ is the so-called \emph{corner}. The portion of
$\mathscr{I}$ in the future of $\mathcal{S}_\star$ will be denoted by
$\mathscr{I}^+$. In addition, it will be assumed that the causal future
$J^+(\mathcal{S}_\star)$ coincides with the future domain of dependence
$D^+(\mathcal{S}_\star\cup \mathscr{I}^+)$ and that $\mathcal{S}_\star \cup
\mathscr{I}^+\approx \mathcal{S}_\star \times [0,1)$ so that, in particular,
$\mathscr{I}^+ \approx \partial \mathcal{S}_\star \times [0,1)$ ---see \Cref{Figure:AdSDomains}.

\subsection{Coordinates}

Let us introduce adapted coordinates $x = (x^\mu)$ such that $\mathcal{S}_\star$ and
$\mathscr{I}$ are given as
\[
\mathcal{S} =\{ x\in \mathbb{R}^3 \; | \;  x^0=0  \}, \qquad
\mathscr{I} =\{ x\in \mathbb{R}^3 \; | \;  x^1=0  \},
\]
so the corner is defined by the condition $x^0 = x^1 = 0$.  Coordinates are
propagated off $\mathcal{S}_\star$ via the generalised wave
coordinated condition \eqref{GeneralisedWaveCoordinateCondition}. Observe that
this can always be locally solved: the expression above provides the value of
the coordinates on $\mathcal{S}_\star$ while their normal derivatives are
obtained from the coordinates are $(x^\mu)$ are independent, that is to say,
the coordinate differentials $\mathbf{d}x^\mu$ must be linearly independent.

\begin{figure}
\begin{center}
\includegraphics[scale=1]{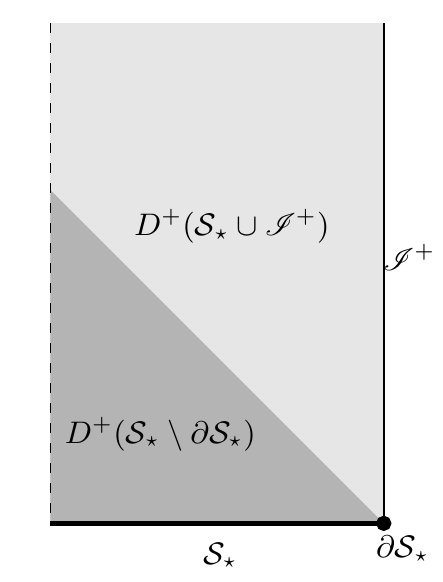}
\end{center}
\caption[Diagram of the construction of anti-de Sitter-like spacetimes.]
{Penrose diagram of the set-up for the construction of anti-de Sitter-like
spacetimes as described in the main text. Initial data prescribed on
$\mathcal{S}_\star\setminus\partial\mathcal{S}_\star$ allows to recover the
dark shaded region $D^+(\mathcal{S}_\star\setminus\partial\mathcal{S}_\star)$.
In order to recover $D^+(\mathcal{S}_\star \cup \mathscr{I}^+)$ it is necessary
to prescribe boundary data on $\mathscr{I}^+$.}
\label{Figure:AdSDomains}
\end{figure}

\subsection{Boundary conditions for the conformal evolution equations}
\label{Section:BoundaryConditions}

In order to formulate an initial-boundary problem for anti-de Sitter-like
spacetimes we must provide Dirichlet boundary data for the wave equations
\eqref{ReducedWaveCFE1}-\eqref{ReducedWaveCFE5} on the conformal boundary. Next we
summarise the results found in \cite{CarVal18b} which directly extend to the
tracefree case.

\medskip
Adopting a Gaussian gauge on the conformal boundary, the metric can be written
as
\[
\bmg \simeq \mathbf{d}x^2 \otimes \mathbf{d}x^2 + \ell_{\alpha\beta}
\mathbf{d}x^\alpha \otimes \mathbf{d}x^\beta, \quad 
(\alpha,\,\beta =0,\, 2,\, 3)
\]
with $\ell_{\alpha\beta}$ the components of the Lorentzian 3-metric
$\ell_{ij}$. In terms of this the boundary data are 
\begin{eqnarray*}
& \Xi \simeq 0,  \qquad s \simeq \notSigma \varkappa(x), \qquad \notL_\alpha \simeq - \notD_\alpha\varkappa, \qquad
\notL_{\alpha\beta} \simeq \notl_{\alpha\beta} -
  \frac12\varkappa^2(x)\ell_{\alpha\beta}, &  \\
& \notL_{11} \simeq \frac16\mathcal{R}(x) - \frac14r + \frac32\varkappa^2(x), \qquad
\notd_{\alpha\beta\gamma} \simeq - \notSigma^{-1}
  y_{\alpha\beta\gamma},& 
\end{eqnarray*}
where $\not\! \Sigma$ is given by $\eqref{SolutionLambda}$. However, the data
for the electric part of the Weyl tensor $\notd_{ij}$ differs from the one in the vacuum
case, so it will be analysed in the next subsection.

\subsubsection{Boundary data for $\not\! d_{ij}$}

As showed in Proposition
\ref{Proposition:ConformalConstraintsConformalBoundary}, the electric part
of the rescaled Weyl tensor is determined by
equation \eqref{DiffEqWeylMatter}.  In order to analyse this differential
equation it is convenient to perform a $2+1$ decomposition on the conformal
boundary. For this purpose consider a foliation of $\mathscr{I}$ given by a
family of spacelike hypersurfaces $\partial\mathcal{S}_t$ with normal vector
$\nu_i$. This induces the projector
\[
s_i{}^j = \ell_i{}^j + \nu_i\nu^j.
\]
In an analogous manner to \Cref{Section:ConformalEinsteinConstraints},
we define the intrinsic derivative operator $\delta_i$ satisfying
$\delta_i s_{jk} = 0$ along with a normal derivative $\delta \equiv
\nu^i \notD_i$. The projector $s_i{}^j$ allows us to further decompose
tensors $\notd_{ij}$ and $\notj_i$ with respect to $\nu^i$ as
\[
\notd_{ij} = w_{ij} -\nu_i w_j -\nu_jw_i + w \nu_i \nu_j, \quad \not\! j_i = p_i - p\nu_i,
\]
 where the different components are defined as: 
\[
w_{ij} \equiv \ell_i{}^k \ell_j{}^l\notd_{kl}, \quad w_i \equiv \nu^l\ell_i{}^k\notd_{kl}, \quad 
w\equiv \nu^i \nu^j \notd_{ij}, \quad p_i \equiv s_i{}^k \notj_k, \quad p \equiv \nu^i \notj_i.
\]
From this it follows that $w_i{}^i = w$.  Assuming that the acceleration $\nu^j
\notD_j \nu_i$ locally vanishes, we introduce the extrinsic curvature $k_{ij}
\equiv \notD_i \nu_j $ associated to the foliation, as well as its trace $k
\equiv s^{ij} k_{ij}$. A calculation shows that equation
\eqref{DiffEqWeylMatter} implies the system
\begin{subequations}
\begin{eqnarray}
&& \delta^i w_i - \delta w = -2 k^{ij}w_{\{ij\}} - \notSigma p, \label{SymmHyp1} \\
&& 2\delta w_i - \delta_i w 
= -2k w_i - 2w_j k_i{}^j + 2\delta^j w_{\{ij\}} + 2\notSigma p_i, \label{SymmHyp2}
\end{eqnarray}
\end{subequations}
with $w_{\{ij\}} \equiv w_{ij} - \frac{1}{2}s_{ij} w$ being the $s$-tracefree
part of $w_{ij}$.

\begin{remark}
\label{BoundaryDataWeyl}
\em{If $w_{\{ij\}}, \ p_i$ and $p$ are given, then equations
\eqref{SymmHyp1}-\eqref{SymmHyp2} constitute a symmetric hyperbolic system for
$w$ and $w_i$. In this sense, these fields constitute additional pieces of
basic boundary data. The particular tracefree matter model in consideration
will determine $p_i$ and $p$ in an explicit way.  On the other hand, as noticed
in \cite{CarVal18b}, the two independent components of $w_{\{ij\}}$ can be
related to the notions of incoming and outgoing radiation. By exploiting the
Newman-Penrose formalism, this can be expressed in terms of two complex-valued
scalar fields
$\psi_0$ and $\psi_4$.}
\end{remark}

The above discussion leads to:

\begin{lemma}
\label{Lemma:BoundaryData}
Let $\mathscr{I}$ be endowed with the following smooth fields:

\begin{itemize}

\item[(i)] a 3-dimensional Lorentzian metric $\ell_{ij}$; 
\item[(ii)] a set of coordinate gauge source functions $\mathcal{F}^{\mu}(x)$
and the gauge function $\varkappa(x)$;
\item[(iii)] a symmetric tensor $w_{\{ij\}}$ which is spatial with respect to
the foliation induced on $\mathscr{I}$ by the functions $\mathcal{F}^{\mu}(x)$
and tracefree with respect to the metric induced on the leaves of the
foliation;
\item[(iv)] a spatial (in the same sense as in (iii)) vector $p_i$, and scalars
$p$ and $\notSigma$;
\item[(v)] a choice of fields $w$ and $w_i$ on a fiduciary hypersurface
$\partial\mathcal{S}_\star$ of $\mathscr{I}$.
\end{itemize}
Then there exists $t_\bullet > 0$ such that on the cylinder
$\mathscr{I}_{t_\bullet} \approx [0,t_\bullet) \times
\partial\mathcal{S}_\star$ the fields $w$ and $w_i$ are uniquely determined
and, together with the fields prescribed in (iii) and (iv), satisfy the
constraint \eqref{DiffEqWeylMatter}.
\end{lemma}

\subsection{Summary}
The analysis of this section can be summarised in the following:

\begin{proposition}
\label{Proposition:SummaryBoundaryData}
Let $(\ell_{ij},\ w_{\{ij\}},\ p_i, \ p)$ be a collection of smooth fields
defined on $\mathscr{I}$ as in \Cref{Lemma:BoundaryData} and $(\notSigma, \ s,
\ \notK_{ij}, \ \notL_i, \ \notL_{ij}, \ \notd_{ijk})$ be the fields
constructed according to the procedure described in the previous subsection.
Then, the zero-quantities defined by relations \eqref{ZQ1}-\eqref{ZQ6} satisfy
\begin{eqnarray*}
& \hspace{-0.7cm} \ell_b{}^a\Theta_a\simeq 0, \quad Z \simeq 0 & \\
& \ell_c{}^a \ell_d{}^b \Upsilon_{ab}\simeq 0, \quad
\notn^a \ell_c{}^b\Upsilon_{ab}\simeq0, & \\
& \ell_e{}^c \ell_f{}^d
\ell_g{}^b \Delta_{cdb}\simeq 0, \quad \notn^b \ell_e{}^c \ell_f{}^d \Delta_{cdb}\simeq 0,& \\
& \notn^b \ell_e{}^c \ell_f{}^d\Lambda_{bcd}\simeq 0, \quad
  \notn^b \notn^d \ell_e{}^c \Lambda_{bcd}\simeq 0, & \\
& \ell_a{}^e \ell_b{}^f \ell_c{}^g \ell_d{}^h P_{efgh} \simeq 0, \quad
\notn^d \ell_a{}^e \ell_b{}^f \ell_c{}^g P_{edfg} \simeq 0, &
\end{eqnarray*}
at least on $\mathscr{I}_{t_\bullet} \approx [0,t_\bullet)\times \partial\mathcal{S}_\star$.

\end{proposition}

\begin{remark}
{\em Notice that the boundary data discussed throughout this section is not
necessarily consistent with the initial data at the corner. Demanding the
compatibility of these two sets of data requires to impose so-called
\emph{corner conditions}. If one considers asymptotically hyperbolic initial
data, a gluing construction allows to satisfy these conditions to any arbitrary
order \cite{ChrDel09}.  In the present construction, this issue can be
addressed by exploiting the conformal constraint equations and the conformal
Einstein field equations leading to a recursive procedure where the $n$-th
order conditions are dependent on the ones at $(n-1)$-th order for $n>1$. For a more
detailed discussion see \cite{LueVal14a,CarVal18b}.}
\end{remark}

\section{Boundary data for the zero-quantities}
\label{Section:PropOfConstr}

\Cref{Proposition:SummaryBoundaryData} establishes that the boundary
conditions discussed in \Cref{Section:BoundaryConditions} represent
the vanishing of a number of components of the zero-quantities on
$\mathscr{I}$ ---specifically, the ones involving intrinsic
derivatives. In order to show that these, in turn, imply the vanishing
of the remaining ones, we make use of their properties along with the
geometric wave equations \eqref{CWE1}-\eqref{CWE4} and
\eqref{EquationMetric} ---see \cite{CarHurVal19} for further
details. In particular, the relevant expressions are:
\begin{subequations}
\begin{eqnarray}
&& P^c{}_{acb} = 0, \label{TraceZQ6} \\
&& \Upsilon_a{}^a = 0, \label{TraceZQ1} \\
&& \nabla_a \Theta^a = \Upsilon^{ab}L_{ab} - \tfrac12 \Xi^2 \Upsilon^{ab}T_{ab}, \label{DivergenceZQ2} \\
&& \nabla_d P_{abc}{}^d = - \Delta_{abc} -\Xi\Lambda_{cab}, \label{DeltaLambdaIdentity} \\
&& \nabla_c \Delta_a{}^c{}_b = \Upsilon^{cd}d_{acbd} + \Lambda_{abc}\nabla^b\Xi
- L^{cd}P_{acbd}, \label{ZQWE3} \\
&& 3\nabla_{[d}\Delta_{ab]c} = \Lambda_{abdce} \nabla^{e}\Xi + 3 \Upsilon_{[a}{}^{e}d_{bd]ce} 
+ 3 L_{[a}{}^{e}P_{bd]ce} -  \tfrac{3}{2} \Xi^2 P_{[ab|c|}{}^{e}T_{d]e}
+ 2 \Xi \Upsilon_{[a}{}^{e}g_{b|c|}T_{d]e} \nonumber \\
&& \hspace{2cm} + \Xi \Upsilon_{[a}{}^{e}g_{|c|b}T_{d]e}, \label{IC3}
\end{eqnarray}
\end{subequations}
where $\Lambda_{dabce} \equiv 3\Lambda_{c[da} g_{b]e} - 3\Lambda_{e[da} g_{b]c}$.

\medskip

\noindent \textbf{Boundary data for $P^a{}_{bcd}$}. By definition, the field
$P^a{}_{bcd}$ inherits the symmetries of the Riemann tensor. This makes
possible to decompose it into three main components: 
\[
\hat{P}_{abcd} \equiv \ell_a{}^e \ell_b{}^f \ell_c{}^g \ell_d{}^h P_{efgh},
\quad \hat{P}_{abc} \equiv \notn^d \ell_a{}^e \ell_b{}^f \ell_c{}^g P_{edfg},
\quad \hat{P}_{ab} \equiv \notn^c \notn^d \ell_a{}^e \ell_b{}^f P_{ecfd}.
\]
The first two vanish by virtue of the constraints
\eqref{BdyConformalConstraint9} and \eqref{BdyConformalConstraint10}, while a
calculation shows that $\hat{P}_{ab} \simeq  P^c{}_{acb} - \notn_a \notn_b
\notn^c \notn^d P^e{}_{ced}$. From equation $\eqref{TraceZQ6}$ it follows that
$\hat{P}_{ab} \simeq 0$.

\medskip

\noindent \textbf{Boundary data for $\Upsilon_{ab}$}.  The zero-quantity
$\Upsilon_{ab}$ can be decomposed with respect to $\notn^a$ by defining the
projections $\gamma_{ab} \equiv \ell_a{}^c \ell_b{}^d \Upsilon_{cd}, \ \gamma_a
\equiv \notn^b \ell_a{}^c \Upsilon_{bc}$ and $\gamma \equiv \notn^a \notn^b
\Upsilon_{ab}$. Accordingly, we can write
\[
\Upsilon_{ab} = \gamma_{ab} + 2\gamma_{(a} \notn_{b)} + \gamma \notn_a \notn_b.
\]
The prescription of the boundary data discussed in the previous section implies
that $\gamma_{ab} \simeq 0$ and $\gamma_a \simeq 0$. On the other hand, writing
equation \eqref{TraceZQ1} in terms of these projections one has that $\gamma
\simeq \gamma_a{}^a \simeq 0$.

\medskip

\noindent \textbf{Boundary data for $\Theta_a$}. Consider the
projections $\theta_a \equiv \ell_a{}^b\Theta_b$ and $\theta
\equiv \notn^a\Theta_a$. Then we have that
\[
\Theta_a = \theta_a + \notn_a\theta.
\]
The boundary data prescription for $\notL_i$ is equivalent to $\theta_a \simeq
0$. In order to prove the vanishing of $\theta$ we use the identity
\eqref{DivergenceZQ2}. Since $\Upsilon_{ab} \simeq 0$, a short calculation
yields
\[
\notD\theta \simeq - 3\theta\varkappa(x).
\]
Without loss of generality, it is always possible, under a further conformal
rescaling of the form $g'_{ab} = \omega^2 g_{ab}$, to choose a conformal
representation for which $\varkappa \simeq 0$ ---see
\Cref{Lemma:ConverseProposition}--- meaning that $\notD\theta \simeq 0$.  Last
equation then implies that $\theta\varkappa(x) \simeq 0$; nevertheless, thanks
to the conformally invariance, this result is valid in general. If $\varkappa =
0$, the above rescaling enables us to find a representation for which
$\varkappa \neq 0$ leading, in any case, to conclude that $\theta \simeq 0$.

\medskip

\noindent \textbf{Boundary data for $\Delta_{abc}$}. Consider the system of
wave equations for the geometric fields
\eqref{ReducedWaveCFE1}-\eqref{ReducedWaveCFE5}. As initial and boundary data
sets for the system have already been established, a solution can then be
locally obtained in a neighbourhood of $\partial\mathcal{S}_\star$. In
particular, $d^a{}_{bcd}$ and its derivatives are well-defined, which means
that all the components of $\Lambda_{abc}$ are regular. Moreover, it
can be checked that the trivial data for $P^a{}_{bcd}$ imply that $\nabla_d
P_{abc}{}^d \simeq 0$.  Thus, from equation \eqref{DeltaLambdaIdentity} we
conclude that $\Delta_{abc} \simeq 0$.

\medskip

\noindent \textbf{Boundary data for $\Lambda_{abc}$}. In the case of
$\Lambda_{abc}$ we introduce the independent components $\lambda_{abc} \equiv
\ell_a{}^d \ell_b{}^e \ell_c{}^f \Lambda_{def}$, $\lambda_{ab} \equiv \notn^c
\ell_a{}^d \ell_b{}^e \Lambda_{cde}$, $\Lambda_{ab} \equiv \notn^c \ell_a{}^d
\ell_b{}^e \Lambda_{dce}$ and $\Lambda_{a} \equiv \notn^b \notn^c \ell_a{}^d
\Lambda_{bcd}$. In terms of these we have that 
\begin{equation}
\Lambda_{abc} = \lambda_{abc} + \lambda_{bc}\notn_a + 2\Lambda_{a[c}\notn_{b]} +
2 \Lambda_{[c}\notn_{b]}\notn_a. \label{DecompositionLambda}
\end{equation}
The boundary data for the electric and magnetic parts of $d^a{}_{bcd}$ are
equivalent to $\lambda_{ab} \simeq 0$ and $\Lambda_a \simeq 0$. Next, 
we proceed to prove that the two remaining components vanish as well.
First, consider the normal derivative of the identity \eqref{DeltaLambdaIdentity} and
project all its free indices onto $\mathscr{I}$. This results in
\[
\notSigma \lambda_{abc} \simeq - \notD\delta_{abc},
\]
where $\delta_{abc} \equiv \ell_a{}^d\ell_b{}^e\ell_c{}^f\Delta_{def}$.
Furthermore, projecting the identity \eqref{IC3} with
$\notn^a\ell_f{}^b\ell_g{}^d\ell_h{}^d$ and using the vanishing of $\Upsilon_{ab}$,
$\Delta_{ab}$ and $P^a{}_{bcd}$ on $\mathscr{I}$, a calculation yields
\[
\notD\delta_{abc} \simeq 0,
\]
which then implies that $\lambda_{abc} \simeq 0$.

To complete the proof, define two further components of $\Delta_{abc}$:
$\Delta_{ab} \equiv \notn^c \ell_a{}^d \ell_b{}^e \Delta_{cde}$ and $\Delta_a
\equiv \notn^b \notn^c \ell_a{}^d \Delta_{bdc}$. Observe that multiplying
\eqref{ZQWE3} by $\notn^c$, one readily finds that $\notD \Delta_a \simeq 0$.
This, in turn, can be used to obtain expressions for the normal derivative of
$\Delta_{ab}$ by means of equations \eqref{DeltaLambdaIdentity} and
\eqref{ZQWE3}, namely
\begin{subequations}
\begin{eqnarray}
&& \notD \Delta_{ab} \simeq \notSigma\Lambda_{ab},\\
&& \notD \Delta_{ab} \simeq -\notSigma\Lambda_{ba}.
\end{eqnarray}
\end{subequations}
From here we have that $\Lambda_{(ab)} \simeq 0$.  On the other hand, by
exploiting the identity $\Lambda_{[abc]}=0$, a simple calculation shows that
$2\Lambda_{[ab]} = \lambda_{ab} \simeq 0$, which proves that $\Lambda_{ab}
\simeq 0$.

\medskip

The above results can be summarised as:

\begin{lemma}
\label{Lemma:ZQBoundaryData}
Assume that the wave equations \eqref{CWE1}-\eqref{CWE4} and
\eqref{EquationMetric} are valid. If the boundary data for the geometric fields
are given as in \Cref{Proposition:SummaryBoundaryData}, then the geometric
zero-quantities vanish on $\mathscr{I}$.
\end{lemma}

\begin{remark}
\label{Remark:ZQInitialData}
{\em The components of the zero-quantities on
  $\partial\mathcal{S}_\star$ 
corresponding to projections on this hypersurface vanish by the way the anti-de
Sitter-like initial data has been constructed. Components with a transversal
(i.e., timelike) projection can be read as a first order evolution system for
the geometric conformal fields. Thus, in order to ensure the vanishing of the
zero-quantities on $\mathcal{S}_\star$, one needs, firstly, to produce a
solution to the conformal constraint equations. Secondly, one reads the
transversal components of the zero-quantities as definitions for the normal
derivatives of the conformal fields which can be readily computed from the
solution to the conformal constraints. In this sense, the transversal
components of the zero-quantities vanish \emph{a fortiori}. Furthermore, as a
consequence of this procedure, the normal derivatives of the zero-quantities
trivially vanish on $\mathcal{S}_\star$.}
\end{remark}

\section{Matter models}
\label{Section:MatterModels}

In this section several specific tracefree models of interest are
studied. Namely, the conformally invariant scalar field, the Maxwell field and the
Yang-Mills field. The problem of the coupling of these matter models to
the system \eqref{ReducedWaveCFE1}-\eqref{ReducedWaveCFE5} has been addressed
in \cite{CarHurVal19}. In particular, a system of wave equations has been
obtained in conjunction with an analysis for the respective subsidiary
variables. The aim of this section is to identify the basic boundary data
corresponding to each matter model required by the systems
\eqref{ReducedWaveCFE1}-\eqref{ReducedWaveCFE5} and
\eqref{SymmHyp1}-\eqref{SymmHyp2}. An investigation of their relation to the
propagation of the constraints on $\mathscr{I}$ is also provided.

\subsection{The conformally invariant scalar field}

The conformally invariant scalar field is a first example of an explicit
tracefree matter model of interest. Let $\tilde{\phi}$ be a scalar field on
$(\tilde{\mathcal{M}}, \tilde{\bmg})$ governed by the equation
\begin{equation*}
\tilde{\nabla}_a \tilde{\nabla}^a \tilde{\phi} - \frac16 \tilde{R} \tilde{\phi} = 0.
\end{equation*}
Defining the unphysical scalar field $\phi \equiv \Xi^{-1}\tilde{\phi}$, it is
this equation remains invariant under the conformal rescaling
\eqref{ConformalRescaling}. This means that $\phi$ satisfies
\begin{equation}
\nabla_a \nabla^a \phi - \frac16 R \phi = 0.
\label{InvariantScalarField}
\end{equation}
Furthermore, the energy-momentum tensor associated to this field is
\begin{equation}
T_{ab} = \nabla_{a}\phi \nabla_{b}\phi -  \tfrac{1}{2} \phi \nabla_{a}\nabla_{b}\phi -  
\tfrac{1}{4} g_{ab} \nabla_{c}\phi \nabla^{c}\phi + \tfrac{1}{2} \phi^2 L_{ab}.
\label{TabScalarField}
\end{equation}

Due to the second order derivatives of $\phi$ in the last expression,
the coupling of 
this matter model to the wave equations
\eqref{ReducedWaveCFE1}-\eqref{ReducedWaveCFE5} will result in the appearance
of up to fourth order derivatives of $\phi$ in the evolution system. This can
be remedied by first noticing that the third order derivatives are of the form
$\nabla_{[a}\nabla_{b]}\nabla_c\phi$ ---see equation
\eqref{CottonTensorDefinition}--- so using the commutator of covariant
derivatives they can be expressed in terms of $\nabla_a\phi$. First and second
derivatives, on the other hand, can be removed by introducing the auxiliary
fields
\begin{equation}
\phi_a \equiv \nabla_a\phi, \quad \phi_{ab} \equiv \nabla_a\nabla_b \phi,
\label{AuxiliaryVariablesScalarField}
\end{equation}
satisfying the system of wave equations
\begin{subequations}
\begin{eqnarray}
&& \square \phi_a = 2\phi^bL_{ab} + \frac13 R\phi_a + \frac16\phi\nabla_a R, \label{WaveEqnDPhi} \\
&& \square\phi_{ab} = \tfrac{1}{2} \phi_{ab} R -  \tfrac{1}{3} R \phi L_{ab} - 2 \phi^{cd} L_{cd}
g_{ab} -  \tfrac{1}{6} \phi^{c} g_{ab} \nabla_{c}R + \tfrac{1}{6} \phi
\nabla_{(a}\nabla_{b)}R - 2 \Xi \phi^{cd}d_{acbd} \nonumber \\
&& \hspace{1.2cm} + 8 \phi_{(a}{}^{c}L_{b)c} + 2 \Xi \phi^{c}T_{(a|c|b)} + \tfrac{2}{3}
\phi_{(a}\nabla_{b)}R + 2 \phi^{c}\nabla_{(a}L_{b)c}
- 2 \phi^{c}d_{(a|c|b)}{}^{d}\nabla_{d}\Xi.\label{WaveEqnDDPhi} 
\end{eqnarray}
\end{subequations}
Writing them in terms of the reduced wave operator, these equations together
with \eqref{InvariantScalarField} must be coupled to the system
\eqref{ReducedWaveCFE1}-\eqref{ReducedWaveCFE5}.

\subsubsection{Basic boundary data}

Next, we analyse the data one must prescribe on $\mathscr{I}$. As the
fields $\phi, \ \phi_a$ and $\phi_{ab}$ satisfy wave equations, it is necessary
to prescribe suitable Dirichlet boundary data for them. First observe that
$\phi$ can be freely prescribed as its value is not constrained by any
equation intrinsic to $\mathscr{I}$. Furthermore, as $\phi$ is governed by the
second order equation \eqref{InvariantScalarField}, then first order
derivatives also constitute a piece of basic data. More specifically, it is
only necessary to establish the values of the normal derivative $\notD\phi$ on
the conformal boundary.

\subsubsection{Boundary data for the evolution systems}

In order to analyse the Dirichlet boundary data for the auxiliary fields it is
convenient to introduce the projections
\[
\varphi_a \equiv \ell_a{}^b \phi_b, \qquad \varphi \equiv \not\! n^a\phi_a, \qquad
\bar{\phi}_{ab} \equiv \ell_a{}^c \ell_b{}^d \phi_{cd}, \qquad \bar{\phi}_a \equiv \not\! n^c \ell_a{}^b
\phi_{bc}.
\]
From the discussion above, $\varphi_a$ and $\varphi$ can be obtained directly
once the basic data have been imposed; these represent the boundary data for
$\phi_a$. On the other hand, observing that $\phi_a{}^a = \tfrac16R\phi$ we can
write
\[
\phi_{ab} = \bar{\phi}_{ab} + \not\! n_a \bar{\phi}_b + \not\! n_b \bar{\phi}_a
+ (\tfrac16 R\phi - \bar{\phi}_a{}^a)\not\!n_a\not\!n_b.
\]
Since $\bar{\phi}_{ab}$ and $\bar{\phi}_a$ can, via commutation of covariant derivatives, be
determined from $\phi$ on the conformal boundary it follows that the boundary
data for $\phi_{ab}$ is completely determined from the basic data.

Finally, we focus on the boundary data $p_i$ and $p$ required for the system
\eqref{SymmHyp1}-\eqref{SymmHyp2}. From the decompositions for
$\phi_a$ and $\phi_{ab}$, along with expression \eqref{TabScalarField}
one has that 
\[
\not\! j_i \simeq \varphi \varphi_{i} - \frac{1}{2} \phi^2 \notD_{i}\varkappa
+ \frac{1}{2} \phi (\varkappa \varphi_i - \notD_i \varphi).
\]
This shows that $\not\! j_i$ can be expressed in terms of the basic boundary
data and, in consequence, the fields $p_i$ and $p$ are computable.

\subsubsection{Boundary data for the subsidiary fields}

In the same spirit of  \Cref{Lemma:WEZQ}, it is now necessary to prove the
consistency of the definitions \eqref{AuxiliaryVariablesScalarField}. To this
end we introduce the subsidiary fields
\begin{subequations}
\begin{eqnarray}
&& Q_a \equiv \phi_a - \nabla_a \phi, \\
&& Q_{ab} \equiv \phi_{ab} - \nabla_a\nabla_b \phi,
\end{eqnarray}
\end{subequations}
where $Q_{ab}$ is symmetric and tracefree.  From the previous discussion, it can
be easily checked that the prescription of boundary data for $\phi$ is
equivalent to
\[
\ell_a{}^b Q_b \simeq 0, \quad \notn^a Q_a \simeq 0, \quad \ell_a{}^c \ell_b{}^d Q_{cd} \simeq 0,
\quad \notn^b \ell_a{}^c Q_{bc} \simeq 0.
\]
Exploiting the fact that $Q_a{}^a =0$, it readily follows that $\notn^a \notn^b
Q_{ab} \simeq -\ell^{ab}Q_{ab} \simeq 0$, implying the vanishing of $Q_a$ and
$Q_{ab}$ on $\mathscr{I}$.

\begin{remark}
{\em Similarly, we prescribe initial data consisting of $\phi$ and $D\phi$ on
$\mathcal{S}_\star$ in an analogous manner as it was done on $\mathscr{I}$. In
consequence, vanishing initial data for $Q_a$ and $Q_{ab}$ are obtained, which
in turn implies that their intrinsic first derivatives vanish. Furthermore, it
can be directly that their normal derivatives vanish too. Hence, $\nabla_a
Q_{b} = 0$ and $\nabla_a Q_{bc} = 0$ on $\mathcal{S}_\star$.}
\end{remark}

On the other hand, assuming the validity of wave equations for $\phi$, $\phi_a$
and $\phi_{ab}$ the subsidiary fields satisfy geometric wave equations. 
Denoting $Q_a$ and $Q_{ab}$ as $\bmQ$ and $\bm{Q'}$, respectively, we have
\begin{eqnarray*}
&& \square Q_{a} = H_a(\bmQ, \bm{Q'}), \\ 
&& \square Q_{ab}= H_{ab}(\bmQ, \bm{Q'}, \bm\Delta),
\end{eqnarray*}
where $H_a$ and $H_{ab}$ are homogeneous functions of their arguments. From the
discussion above, along with \Cref{Lemma:ZQBoundaryData} and
\Cref{Remark:ZQInitialData}, this system can be supplemented with suitable
vanishing initial-boundary data. In consequence, the unique solution on the
spacetime is the trivial one, i.e. $Q_a = Q_{ab} = 0$, proving the consistency of
definitions \eqref{AuxiliaryVariablesScalarField}.

\subsubsection{Summary}

The material of this subsection can be summed up as:

\begin{lemma}
\label{Lemma:ScalarField}
Let $\phi$ be the conformally invariant scalar field satisfying equation
\eqref{InvariantScalarField} with energy-momentum tensor given by
\eqref{TabScalarField}. If $\phi$ and its normal derivative are prescribed on
$\mathcal{S}_\star$ and $\mathscr{I}$, then the system
\eqref{ReducedWaveCFE5}-\eqref{ReducedWaveCFE5} coupled to
\eqref{InvariantScalarField}, \eqref{WaveEqnDPhi} and \eqref{WaveEqnDDPhi},
written in terms of the reduced wave operator, constitute a proper system of
quasilinear wave equations for the Einstein-conformally invariant scalar field
system.
\end{lemma}

\subsection{The Maxwell field}

The next example under consideration is the electromagnetic field. In the physical
spacetime the information is encoded in the antisymmetric Faraday
tensor $\tilde{F}_{ab}$ which satisfies Maxwell equations
\begin{eqnarray*}
&& \tilde{\nabla}^a \tilde{F}_{ab} = 0, \\
&& \tilde{\nabla}_{[a} \tilde{F}_{bc]} = 0.
\end{eqnarray*}
Defining the rescaling $F_{ab} \equiv \tilde{F}_{ab}$, the unphysical Maxwell
equations have the same form, namely
\begin{subequations}
\begin{eqnarray}
&& \nabla^a F_{ab} = 0, \label{MaxwellEqn1} \\
&& \nabla_{[a} F_{bc]} = 0. \label{MaxwellEqn2}
\end{eqnarray}
\end{subequations}
In terms of the Hodge dual $F^*_{ab} \equiv
-\frac12\epsilon_{ab}{}^{cd}F_{cd}$, equation \eqref{MaxwellEqn2} can be
written, alternatively, as
\begin{equation}
\nabla_b F^*_{ab} = 0. \label{MaxwellEqn2Dual}
\end{equation}
Also, the energy-momentum tensor of the Maxwell field takes the form
\begin{equation}
T_{ab} = F_{ac}F_b{}^c - \frac14 g_{ab}F_{cd}F^{cd},
\label{TabFaraday}
\end{equation}
in agreement with conservation law \eqref{UnphysicalConservationLaw}.

From equation \eqref{TabFaraday} is clear that the coupling of the Maxwell
field to the system \eqref{ReducedWaveCFE1}-\eqref{ReducedWaveCFE5} results in
the appearance of second order derivatives of $F_{ab}$. The hyperbolicity of
the system can be restored adopting a similar strategy as for the conformally
invariant scalar field. For this purpose we introduce the fully tracefree
tensor field
\begin{equation}
F_{abc} \equiv \nabla_a F_{bc}.
\label{DefinitionFabc}
\end{equation}
Formulae \eqref{MaxwellEqn1}-\eqref{MaxwellEqn2} imply that $F_{ab}$ and $F_{abc}$
satisfy the following system of geometric wave equations:
\begin{subequations}
\begin{eqnarray}
&& \square F_{ab}=\tfrac{1}{3}F_{ab}R-2\Xi F^{cd}d_{acbd}, \label{WaveEquationFaraday} \\
&& \square F_{abc}  = -2\Xi F_{a}{}^{d} T_{bcd} + 4\Xi F_{[b}{}^{d} T_{|ad|c]} 
- 2 \Xi F_{a}{}^{de} d_{bdce} - 4 \Xi F^{d}{}_{[b}{}^{e} d_{c]ead} 
+ \tfrac{1}{2} F_{abc} R  + 4 F^{d}{}_{bc} L_{ad} \nonumber \\
&& \hspace{1.4cm}  - 4 F^{d}{}_{a[b} L_{c]d}
- 4 F^{d}{}_{[b}{}^{e} g_{c]a} L_{de} + \tfrac{1}{3} F_{bc} \nabla_{a}R 
- 2 F^{de} d_{ade[b}\nabla_{c]}\Xi - 4 \Xi F^{de} \nabla_{[b}d_{c]ead} \nonumber \\
&& \hspace{1.4cm} - \tfrac{1}{3} F_{a[b} \nabla_{c]}R 
- 2 F_{[b}{}^{e} d_{c]ead} \nabla^{d}\Xi -  F_{d}{}^{e} d_{aebc} \nabla^{d}\Xi 
- 4 F_{[b}{}^{e} d_{c]dae} \nabla^{d}\Xi - F_{a}{}^{e} d_{bcde} \nabla^{d}\Xi \nonumber \\
&& \hspace{1.4cm} + 2 F^{ef} g_{a[b} d_{c]edf}\nabla^{d}\Xi + \tfrac{1}{3} g_{a[b} F_{c]d} \nabla^{d}R.
\label{WaveEquationDerFaraday}
\end{eqnarray}
\end{subequations}
Replacing $\square$ by the reduced wave operator, these are cast as hyperbolic
wave equations and can be coupled to the system
\eqref{ReducedWaveCFE1}-\eqref{ReducedWaveCFE5}.

\subsubsection{Basic boundary data}

The Faraday tensor accepts a simple decomposition with respect to $\notn^a$ in
terms of its electric and magnetic parts defined, respectively, as $F_a \equiv
\not\! n^c \ell_a{}^b F_{bc}$ and $F^*_a \equiv \not\! n^c \ell_a{}^b
F^*_{bc}$, namely
\begin{equation}
F_{ab} = F_a \not\!n_b - F_b \notn_a + \epsilon^c{}_{ab}F^*_c, \qquad
F^*_{ab} = F^*_a \not\!n_b - F^*_b \not\!n_a - \epsilon^c{}_{ab}F_c,
\label{FaradayDecompositions}
\end{equation}
where $\epsilon_{abc} \equiv \notn^d \epsilon_{adbc}$ is the 3-volume form
induced on $\mathscr{I}$. Unlike the conformally invariant scalar field, these
components cannot be freely prescribed, but Maxwell equations impose
a set of constraints on $\mathscr{I}$:
\begin{subequations}
\begin{eqnarray}
&& \not\!\! D^i F_i \simeq 0, \label{GaussEqn1} \\
&& \not\!\! D^iF^*_i \simeq 0. \label{GaussEqn2}
\end{eqnarray}
\end{subequations}
In order to identify the basic data for the Maxwell field we perform a further
decomposition with respect to $\nu^i$. Introducing the
projections $f_i \equiv s_i{}^jF_j, \ f \equiv \nu^iF_i, \ f^*_i \equiv
s_i{}^jF^*_j$ and $f^* \equiv \nu^i F^*_i$, these take the form
\begin{subequations}
\begin{eqnarray}
&& \delta f - kf - \delta^i f_i \simeq 0, \label{MaxwellConstraint1} \\
&& \delta f^* - kf^* - \delta^i f^*_i \simeq 0, \label{MaxwellConstraint2}
\end{eqnarray}
\end{subequations}
which represent an evolution system for $f$ and $f^*$.  Observe that the data
required to establish a well-posed problem for this system correspond to the
fields $f_i$ and $f^*_i$ on $\mathscr{I}$ along with initial conditions for $f$
and $f^*$ at $\partial\mathcal{S}_\star$.

\subsubsection{Boundary data for the evolution systems}

Having identified the basic data for the Maxwell field, we now proceed to
verify that $\not\! j_i$ can be determined in terms of the basic data. A
calculation using equation \eqref{TabFaraday} yields
\begin{equation}
\not\! j_ i = -\epsilon_{ijk}F^{j}F^{*k}.
\label{PoyntingFaraday}
\end{equation}
Directly from the definitions of $p_i$ and $p$ it follows that 
\begin{equation}
p_i = \epsilon_{ij} (f^* f^j - f f^{*j}), \qquad
p = -\epsilon_{jk}f^jf^{*k}.
\end{equation}
where $\epsilon_{ij} \equiv \nu^k\epsilon_{kij}$ is the 2-dimensional volume
form on the foliation $\partial\mathcal{S}_t$ satisfying $\nu^i\epsilon_{ij}
=0$. Here $f$ and $f^*$ can be obtained from evolving equations
\eqref{MaxwellConstraint1}-\eqref{MaxwellConstraint2}. Therefore, the boundary
data for the system \eqref{SymmHyp1}-\eqref{SymmHyp2} are completely
determined.

Next we show that the basic data also provides data for $F_{abc}$. For this
purpose, we introduce a number of components with respect to $\notn^a$:
\begin{eqnarray*}
& f_{abc} \equiv \ell_{a}{}^d \ell_{b}{}^e \ell_{c}{}^f F_{def}, \quad
f_{ab} \equiv \notn^d \ell_a{}^e \ell_b{}^f F_{def}, \quad \hat{f}_{ab} \equiv \notn^e \ell_a{}^d \ell_b{}^f F_{def},
\quad F_a \equiv \notn^b \notn^c \ell_a{}^d F_{bcd}.
\end{eqnarray*}
As $F_{abc}$ possesses the same symmetries as $\Lambda_{abc}$ ---see equation
\eqref{DecompositionLambda}--- we can write
\begin{equation}
F_{abc} = f_{abc} + f_{bc}\notn_a + 2\hat{f}_{a[c}\notn_{b]} + 2F_{[c}\notn_{b]}\notn_a.
\label{DerFaradayDecomposition}
\end{equation}
From their definitions, a series of straightforward calculations result in
\begin{subequations}
\begin{eqnarray}
&& f_{abc} \simeq \varkappa(x)f_{[b}\ell_{c]a} + \varkappa(x)\epsilon_{defg}\ell_a{}^e \ell_b{}^f
\ell_c{}^g f^{*d} + \epsilon_{egh}\ell_b{}^g \ell_c{}^h \notD_a f^{*e}, \\
&& \hat{f}_{ab} \simeq - \varkappa(x)\epsilon_{dfc}\ell_a{}^f \ell_b{}^c f^{*d} - \notD_a f_b.
\end{eqnarray}
\end{subequations}
On the other hand, one can exploit the symmetries of $F_{abc}$ to obtain
expressions for the remaining components. In particular, using that $F_{[abc]}
= 0$ and $F_{ab}{}^a = 0$, along with decomposition
\eqref{DerFaradayDecomposition}, we obtain
\begin{subequations}
\begin{eqnarray}
&& f_{ab} \simeq 2\hat{f}_{[ab]}, \\
&& F_a \simeq f_{ba}{}^b.
\end{eqnarray}
\end{subequations}
Thus, once basic data for $f^a$ and $f^{*a}$ have been provided on
$\mathscr{I}$, all the boundary data for the field $F_{abc}$ are computable.

\subsubsection{Boundary data for the subsidiary fields}

In the next step of the analysis, we are now required to prove that any
solution to the wave equations
\eqref{WaveEquationFaraday}-\eqref{WaveEquationDerFaraday} is also a solution
to the unphysical Maxwell Equations. Accordingly, we define the subsidiary
fields
\begin{subequations}
\begin{eqnarray}
&& M_a \equiv \nabla^b F_{ab},\label{DefinitionZQMaxwell1} \\
&& M_{abc}\equiv \nabla_{[a}F_{bc]}, \label{DefinitionZQMaxwell2} \\
&& Q_{abc} \equiv F_{abc} - \nabla_a F_{bc}. \label{DefinitionZQMaxwellAuxiliary} 
\end{eqnarray}
\end{subequations}
Boundary data for the fields obtained from the constraints
\eqref{GaussEqn1}-\eqref{GaussEqn2} is equivalent to
\[
\notn^a M_a \simeq 0, \qquad \ell_a{}^d \ell_b{}^e \ell_c{}^f M_{def} \simeq 0.
\]
Additionally, we also have that $Q_{abc} \simeq 0$ as a consequence from the way
the data for $F_{abc}$ were constructed. The vanishing of the remaining boundary data
for the subsidiary fields follows from observing that 
\begin{subequations}
\begin{eqnarray}
&& \ell_a{}^b M_b = \ell_a{}^b (F^c{}_{cb} - Q^c{}_{cb})
\simeq f^c{}_{ca} + F_a \simeq 0, \\
&& \notn^c\ell_a{}^d \ell_b{}^e M_{cde} = \notn^c\ell_a{}^d \ell_b{}^e (F_{[cde]} - Q_{[cde]})
\simeq \tfrac13 (f_{ab} + 2\hat{f}_{[ba]}) \simeq 0.
\end{eqnarray}
\end{subequations}

\begin{remark}
\label{Remark:FaradayInitialData}
{\em Following a similar argument, the constraints on $\mathcal{S}_\star$
establish basic initial data which implies the vanishing of the subsidiary
variables involving intrinsic derivatives, as well as of $Q_{abc}$. In the same
vein as in \Cref{Remark:ZQInitialData}, Maxwell equations additionally provide
evolution equations for the corresponding electric and magnetic fields. Taking
these as definitions for the normal derivatives, the remaining components of
$M_a$ and $M_{abc}$ trivially vanish on $\mathcal{S}_\star$. In this regard, a
solution to the Maxwell equations must first be obtained. Under this
construction, the first derivatives of the subsidiary fields vanish on
$\mathcal{S}_\star$.}
\end{remark}

The propagation of the constraints is proven with the help of the system of
geometric wave equations the subsidiary variables satisfy. Using an analogous
notation as in the previous subsection, and representing $M_a$ and $M_{abc}$ by
$\bmM$ and $\bm{M'}$, respectively, we have that
\begin{eqnarray*}
&& \square M_{a} = H_a(\bmM), \\
&& \square M_{abc} = H_{abc}(\bm{M'}), \\
&& \square Q_{abc} = L_{abc}(\bmM, \bm{M'}, \bmQ, \bm\Lambda).
\end{eqnarray*}
As conditions for the vanishing of the initial and boundary data for the subsidiary fields have already
been established, \Cref{Lemma:ZQBoundaryData} and \Cref{Remark:ZQInitialData}
help us to show that the homogeneity of the above system implies that the only
solution on the spacetime corresponds to $M_a=0$, $M_{abc} = 0$ and
$Q_{abc}=0$. In other words, the Maxwell equations are satisfied.

\medskip

The discussion of this subsection can be summarised in the following
statement:
\begin{lemma}
\label{Lemma:MaxwellField}
Let $F_{ab}$ be the Faraday tensor satisfying the Maxwell equations
\eqref{MaxwellEqn1}-\eqref{MaxwellEqn2} with energy-momentum tensor given by
\eqref{TabFaraday}. If the fields $f_i, \ f^*_i$ are prescribed on $\mathscr{I}$
together with $f$ and $f^*$ at $\partial\mathcal{S}_\star$, then the system
\eqref{ReducedWaveCFE1}-\eqref{ReducedWaveCFE5} coupled to
\eqref{WaveEquationFaraday}-\eqref{WaveEquationDerFaraday}, written in terms of
the reduced wave operator, constitute a proper system of quasilinear wave
equations for the Einstein-Maxwell system.
\end{lemma}

\subsection{The Yang-Mills field}

As a third and last example of a tracefree matter field we consider the
Yang-Mills field. Due to its similarities with the Maxwell field, the
analysis is, in great measure, analogous to the one in the previous section.
Let $\mathfrak{g}$ be the Lie algebra of a group $\mathfrak{G}$. The physical
Yang-Mills field consists of a set of fields $\tilde{F}^\fraka{}_{ab}$ and
gauge potentials $\tilde{A}^\fraka{}_a$, where the indices $\fraka, \ \frakb,
\dots$ take values in $\mathfrak{g}$. Let
$C^\fraka{}_{\frakb\frakc}=C^\fraka{}_{[\frakb\frakc]}$ be the structure
constants of $\mathfrak{g}$. The physical Yang-Mills equations are
\begin{eqnarray*}
&& \tilde{\nabla}_a \tilde{A}^\fraka{}_b - \tilde{\nabla}_b \tilde{A}^\fraka{}_a +
C^\fraka{}_{\frakb\frakc} \tilde{A}^\frakb{}_a \tilde{A}^\frakc{}_b -\tilde{F}^\fraka{}_{ab} =0, \\
&& \tilde{\nabla}^a \tilde{F}^\fraka{}_{ab} + C^\fraka{}_{\frakb\frakc}
\tilde{A}^{\frakb a} \tilde{F}^\frakc{}_{ab}=0, \\
&& \tilde{\nabla}_{[a} \tilde{F}{}^\fraka{}_{bc]} + C^\fraka{}_{\frakb\frakc}
\tilde{A}^\frakb{}_{[a} \tilde{F}^\frakc{}_{bc]} =0.
\end{eqnarray*}
Defining the rescaled fields $F^\fraka{}_{ab} \equiv \tilde{F}^\fraka{}_{ab}$
and $A^\fraka{}_{a} \equiv \tilde{A}^\fraka{}_{a}$, the unphysical Yang-Mills
equations are:
\begin{subequations}
\begin{eqnarray}
&& \nabla_a A^\fraka{}_b - \nabla_b A^\fraka{}_a +
C^\fraka{}_{\frakb\frakc} A^\frakb{}_a A^\frakc{}_b -F^\fraka{}_{ab} =0, \label{UnphysicalYM1} \\
&& \nabla^a F^\fraka{}_{ab} + C^\fraka{}_{\frakb\frakc}
A^{\frakb a} F^\frakc{}_{ab}=0, \label{UnphysicalYM2} \\
&& \nabla_{[a} F^{\fraka}{}_{bc]} + C^\fraka{}_{\frakb\frakc}A^{\frakb}{}_{[a} F^{\frakc}{}_{bc]} = 0.
\label{UnphysicalYM3}
\end{eqnarray}
\end{subequations}
By defining the Hodge dual of the strength field $F^{*\fraka}{}_{ab} \equiv -
\tfrac12\epsilon_{ab}{}^{cd}F_{cd}$, relation \eqref{UnphysicalYM3} can
be equivalently expressed as
\begin{equation*}
\nabla^b F^{*\fraka}{}_{ba} + C^\fraka{}_{\frakb\frakc}A^{\frakb a} F^{*\frakc}{}_{ab} = 0.
\end{equation*}

As the associated energy-momentum tensor is given by
\begin{equation}
T_{ab} = \tfrac{1}{4} \delta_{\fraka\frakb}F^\fraka{}_{cd}
F^{\frakb cd} g_{ab}
-\delta_{\fraka\frakb}F^\fraka{}_{ac}F^\frakb{}_b{}^c, \label{EMTensorYM}
\end{equation}
undesired second order derivatives of $F^\fraka{}_{ab}$ will emerge in the
system \eqref{ReducedWaveCFE1}-\eqref{ReducedWaveCFE4}.  Inspired by the
analysis of the Maxwell field, we can deal with this problem by defining the
auxiliary field
\begin{equation}
F^\fraka{}_{abc} \equiv \nabla_a F^\fraka{}_{bc} + C^\fraka{}_{\frakb\frakc}A^\frakb{}_a
F^\frakc{}_{bc}.
\label{DefinitionFabcYM}
\end{equation}
The motivation behind the addition of the term containing the structure
constants will become apparent when the boundary data for the subsidiary
variables are computed. 

\begin{remark}
{\em One key aspect that characterises the coupling of the Yang-Mills to the
main system of geometric fields is the fact that, in order to carry out a
hyperbolic reduction, we require to incorporate a set of gauge source functions
\begin{equation}
f^\fraka (x) \equiv \nabla^a A^\fraka{}_a.
\label{Definition:YMGaugeSourceFunction}
\end{equation}}
\end{remark}

\noindent In terms of this gauge quantity, the fields $F^\fraka{}_{ab}$ and
$F^\fraka{}_{abc}$ satisfy a system of geometric wave equations, namely
\begin{subequations}
\begin{eqnarray}
&& \hspace{-0.7cm} \square A^{\fraka}{}_{a} = \tfrac{1}{6} A^{\fraka}{}_{a} R + 2 A^{\fraka b} L_{ab} 
+ C^{\fraka}{}_{\frakb \frakc} F^{\frakc}{}_{ab} A^{\frakb b}  
+ C^{\fraka}{}_{\frakb \frakc} f^\frakc(x) A^{\frakb}{}_{a}  
- C^{\fraka}{}_{\frakb \frakc} A^{\frakb b} \nabla_{b}A^{\frakc}{}_{a}
+ \nabla_{a}f^{\fraka}(x), \label{CWEPotentialYM} \\
&& \hspace{-0.7cm} \square F^{\fraka}{}_{ab} = -2 \Xi F^{\fraka cd}
d_{acbd} + \tfrac{1}{3} F^{\fraka}{}_{ab} R + 2C^\fraka{}_{\frakb\frakc} F^\frakb{}_a{}^c F^\frakc{}_{bc} 
- 2C^\fraka{}_{\frakb\frakc} F^\frakc{}_{cab}A^{\frakb c}
+ C^\fraka{}_{\frakb\frakc} f^\frakb(x) F^\frakc{}_{ab} \nonumber \\
&& \hspace{0.8cm} - C^\fraka{}_{\frakb\frake} C^\frake{}_{\frakc\frakd} F^\frakd{}_{ab}A^{\frakb c} A^\frakc{}_c.
\label{CWEFaradayYM}\\
&& \hspace{-0.7cm} \square F^\fraka{}_{abc} = \tfrac{1}{2} F^{\mathfrak{a}}{}_{abc} R + 4 F^{\mathfrak{a} d}{}_{bc} L_{ad} 
+ 2 F^{\mathfrak{b} d}{}_{bc} F^{\mathfrak{c}}{}_{ad} C^{\mathfrak{a}}{}_{\mathfrak{b}
\mathfrak{c}} -  F^{\mathfrak{c}}{}_{abc} f^{\mathfrak{b}}(x)
C^{\mathfrak{a}}{}_{\mathfrak{b} \mathfrak{c}} + \tfrac{1}{3} F^{\mathfrak{a}}{}_{bc}\nabla_{a}R \nonumber \\
&& \hspace{0.9cm} - F^{\mathfrak{d}}{}_{abc} A^{\mathfrak{b} d} A^{\mathfrak{c}}{}_{d} C^{\mathfrak{a}}{}_{\mathfrak{b} \frake}
C^{\frake}{}_{\mathfrak{c} \mathfrak{d}}  - 2 A^{\mathfrak{b} d} C^{\mathfrak{a}}{}_{\mathfrak{b} \mathfrak{c}}
\nabla_{d}F^{\mathfrak{c}}{}_{abc} - F^{\mathfrak{a}}{}_{d}{}^{e} d_{aebc} \nabla^{d}\Xi
 - F^{\mathfrak{a}}{}_{a}{}^{e} d_{bcde} \nabla^{d}\Xi \nonumber \\
&& \hspace{0.9cm} + 2 \Xi F^{\mathfrak{a} de} \nabla_{e}d_{adbc}
- 4 \Xi F^{\mathfrak{a} d}{}_{[b}{}^{e}d_{|ad|c]e} - 2 \Xi F^{\mathfrak{a}}{}_{a}{}^{de}d_{[b|d|c]e}
- 4 F^{\mathfrak{a} d}{}_{a[b}L_{c]d} + 4 \Xi F^{\mathfrak{a}}{}_{[b}{}^{d}T_{c]da} \nonumber \\
&& \hspace{0.9cm} + 4 \Xi F^{\mathfrak{a}}{}_{[b}{}^{d}T_{|ad|c]} - \tfrac{1}{3}F^{\mathfrak{a}}{}_{a[b}\nabla_{c]}R 
+ 4 F^{\mathfrak{b}}{}_{a[b}{}^{d}F^{\mathfrak{c}}{}_{c]d}C^{\mathfrak{a}}{}_{\mathfrak{b}
\mathfrak{c}} - 4 F^{\mathfrak{a} d}{}_{[b}{}^{e}L_{|de}g_{a|c]} \nonumber \\
&& \hspace{0.9cm} - 2 \Xi F^{\mathfrak{a} de}T_{[b|de}g_{a|c]} 
- 4 F^{\mathfrak{a}}{}_{[b}{}^{d}d_{|ad|c]}{}^{e}\nabla_{e}\Xi 
- 2 F^{\mathfrak{a}}{}_{[b}{}^{d}d_{|a|}{}^{e}{}_{c]d}\nabla_{e}\Xi
+ 2 F^{\mathfrak{a} de}d_{ad[b|e|}\nabla_{c]}\Xi \nonumber \\
&& \hspace{0.9cm} - \tfrac{1}{3} F^{\mathfrak{a}}{}_{[b}{}^{d}g_{|a|c]}\nabla_{d}R 
- 2 F^{\mathfrak{a}de}g_{a[b}\nabla_{|d|}L_{c]e}.
\label{CWEDerFaradayYM}
\end{eqnarray}
\end{subequations}
Again, one must express these equations in terms of the reduced wave operator
when they are coupled the system
\eqref{ReducedWaveCFE1}-\eqref{ReducedWaveCFE5}.

\subsubsection{Basic boundary data}

Due to the presence of the gauge potentials $A^\fraka{}_a$ in the Yang-Mills
equations, which are coupled to the strength potentials, the identification of
its basic data will require a bit more of work. Nevertheless, the approach to
be adopted will be the same as for the Maxwell field. First, introducing the
projections $F^\fraka{}_a \equiv \notn^c\ell_a{}^b F^\fraka{}_{bc}$ and
$F^\fraka{}_a \equiv \notn^c\ell_a{}^b F^{*\fraka}{}_{bc}$, the fields
$F^{\fraka}{}_{ab}$ and $F^{*\fraka}{}_{ab}$ accept decompositions that are
completely analogous to the ones in \eqref{FaradayDecompositions}. On the other
hand, for the gauge potentials we define $\mathcal{A}^\fraka{}_{a} \equiv
\ell_a{}^b A^\fraka{}_b$ and $\mathcal{A}^\fraka \equiv \notn^a A^\fraka{}_a$.
Accordingly, we have 
\begin{equation}
A^\fraka{}_a = \mathcal{A}^\fraka{}_a + \notn^a \mathcal{A}^\fraka.
\end{equation}
Equations \eqref{UnphysicalYM1}-\eqref{UnphysicalYM3}, provide a set of
relations that are intrinsic to the conformal boundary. The projections defined
above enable us to write them as follows:
\begin{subequations}
\begin{eqnarray}
&&\notD^i F^\fraka{}_i \simeq C^\fraka{}_{\frakb\frakc}\mathcal{A}^\frakc{}_i F^{\frakb i}, 
\label{UnphysicalYM1Decomp1}\\
&& \notD^i F^{*\fraka}{}_i \simeq C^\fraka{}_{\frakb\frakc}\mathcal{A}^\frakc{}_i F^{*\frakb i}, 
\label{UnphysicalYM2Decomp1} \\
&& \notD_i \mathcal{A}^\fraka{}_j - \notD_j\mathcal{A}^\fraka{}_i \simeq 
\epsilon^m{}_{kl}\ell_i{}^k \ell_j{}^l F^{*\fraka}{}_m - C^\fraka{}_{\frakb\frakc}
\mathcal{A}^\frakb{}_i \mathcal{A}^\frakc{}_j. \label{UnphysicalYM3Decomp1}
\end{eqnarray}
\end{subequations}
This system is supplemented with the corresponding decomposition 
of equation \eqref{Definition:YMGaugeSourceFunction}, namely
\begin{equation}
\notD^i \mathcal{A}^\fraka{}_i \simeq f^\fraka(x) - 3\varkappa(x)\mathcal{A}^\fraka - \notD\mathcal{A}^\fraka.
\label{YMGaugeSourceFunctionDecomp1}
\end{equation}
The last expression makes evident that the derivatives of the components of
$A^\fraka{}_a$ are not independent of each other. In particular, this suggests
that $\notD\mathcal{A}^\fraka$ is a piece of the basic data to be prescribed on
$\mathscr{I}$.

Next, a further decomposition on $\mathscr{I}$ with respect to the
timelike vector $\nu^a$ can be carried out. For this purpose we define a number of
objects: $f^\fraka{}_i \equiv \ell_i{}^j F^\fraka{}_j$, $f^\fraka \equiv \nu^i
F^\fraka{}_i$, $a^\fraka{}_i \equiv \ell_i{}^j \mathcal{A}_j$ and $a^\fraka
\equiv \nu^i \mathcal{A}^\fraka{}_i$. It turns out that after suitable
contractions, equations
\eqref{UnphysicalYM1Decomp1}-\eqref{UnphysicalYM3Decomp1} and
\eqref{YMGaugeSourceFunctionDecomp1} produce the intrinsic relations 
\begin{subequations}
\begin{eqnarray}
&& \delta f^\fraka \simeq C^\fraka{}_{\frakb\frakc}(a^\frakc f^\frakb - a^\frakc{}_i f^{\frakb i})
+ kf^\fraka + \delta^i f^\fraka{}_i, \label{SymmHypYM1} \\
&& \delta f^{*\fraka} \simeq C^\fraka{}_{\frakb\frakc}(a^\frakc f^{*\frakb} - a^\frakc{}_i f^{*\frakb i})
+ kf^\fraka + \delta^i f^{*\fraka}{}_i, \label{SymmHypYM2} \\
&& \delta a^\fraka{}_i - \delta_i a^\fraka \simeq C^\fraka{}_{\frakb\frakc}a^\frakb a^\frakc{}_i 
+ 2\epsilon_i{}^j f^{*\fraka}{}_j, \label{SymmHypYM3} \\
&& \delta^i a^\fraka{}_i - \delta a^\fraka \simeq ka^\fraka + f^\fraka(x) - 3\varkappa(x)\mathcal{A}^\fraka -
\notD \mathcal{A}^\fraka. \label{SymmHypYM4}
\end{eqnarray}
\end{subequations}
Prescribing the fields $f^\fraka{}_i$, $f^{*\fraka}{}_i$, $f^\fraka(x)$,
$\mathcal{A}^\fraka$ and $\notD\mathcal{A}^\fraka$ on the conformal boundary,
this constitutes a symmetric hyperbolic system for the fields $f^\fraka$,
$f^{*\fraka}$, $a^\fraka{}_i$ and $a^\fraka$ provided that initial values for
them are given at $\partial\mathcal{S}_\star$.

\subsubsection{Boundary data for the evolution equations}

In view of the fact that $F^\fraka{}_{ab}$ has the same symmetries as its
counterpart $F_{ab}$, it is clear that the energy flux can be calculated in an
analogous fashion. This results in the relation 
\begin{equation}
\notj_i = -\epsilon_{ijk}F^{\fraka j} F^{*\frakb k} \delta_{\fraka\frakb},
\end{equation}
from where the boundary data $p_i$ and $p$ can be directly computed. The data
for $A^\fraka{}_a$, on the other hand, can be extracted from solving the system
\eqref{SymmHypYM1}-\eqref{SymmHypYM4} Regarding $F^\fraka{}_{abc}$, its
symmetries make it possible to perform a decomposition similar to the one
carried out in \eqref{DerFaradayDecomposition}. Ultimately, this shows that the
basic data on $\mathscr{I}$ allow us to determine $F^\fraka{}_{abc}$.

\subsubsection{Data for the subsidiary fields}

The final piece in the analysis of this matter field corresponds to proving that
the basic data on $\mathscr{I}$ implies vanishing data for the relevant
subsidiary fields. The system
\eqref{UnphysicalYM1Decomp1}-\eqref{UnphysicalYM3Decomp1} represents the
relations
\[
\notn^a M^\fraka{}_a \simeq 0, \qquad  \ell_a{}^d \ell_b{}^e \ell_c{}^f M^\fraka{}_{def} \simeq 0,
\qquad \ell_a{}^c \ell_b{}^d M^\fraka{}_{cd} \simeq 0.
\]
Analogous to the Maxwell field, the construction of the data for
$F^\fraka{}_{abc}$ implies that $Q^\fraka{}_{abc} \simeq 0$.  The parallelism
between the Yang-Mills and Maxwell fields makes clear that the components
$\ell_{a}{}^b M^\fraka{}_b$ and $\notn^c \ell_a{}^d \ell_b{}^e
M^\fraka{}_{cde}$ vanish on the conformal boundary. Concerning $M^\fraka{}_{ab}$,
it is possible to show that it satisfies the identity
\[
\tfrac12 C^\fraka{}_{\frakb\frakc} F^{\frakb ab} M^\frakc{}_{ab} - C^\fraka{}_{\frakb\frakc}
A^{\frakb a} M^\frakc{}_a = 0.
\]
As it has been argued that $M^\fraka{}_a \simeq 0$, then it follows that for an
arbitrary field $F^\fraka{}_{ab}$ the condition $M^\fraka{}_{ab} \simeq 0$ must
be satisfied. For completeness, the consistency of the gauge can be proved by
defining a further subsidiary field: $P^\fraka \equiv \nabla^a A^\fraka{}_a -
f^\fraka(x)$. Trivially, equation \eqref{YMGaugeSourceFunctionDecomp1} implies
that $P^\fraka \simeq 0$.

\begin{remark}
{\em Vanishing data for the subsidiary variables and their first derivatives on
$\mathcal{S}_\star$ follow from an argument similar to the one in
\Cref{Remark:FaradayInitialData}.}
\end{remark}

The subsidiary variables associated to the Yang-Mills field satisfy a set of
homogeneous geometric wave equations. Adopting similar notation for function
which are for homogeneous in their arguments as in the two previous sections,
one can write
\begin{eqnarray*}
&& \square M^\fraka{}_a = H^\fraka{}_a (\bmM, \bm\nabla\bmM, \bm{M'}, \bm{M''}, \bmQ, \bm\nabla\bmQ), \\
&& \square M^\fraka{}_{ab} = H^\fraka{}_{ab}(\bm{M'}, \bm{M'}, \bm\nabla\bm{M'}, \bm{M''}, \bmQ),\\
&& \square M^\fraka{}_{abc} = H^\fraka{}_{abc}(\bm{M'}, \bm{M''}, \bm\nabla\bm{M''}, \bmQ, \bm\nabla\bmQ),\\
&& \square Q^\fraka{}_{abc} = L^\fraka{}_{abc}(\bmM, \bm{M'}, \bm{M''}, \bmQ, \bm\nabla\bmQ, \bm\Lambda),\\
&& \square P^\fraka = H^\fraka(\bmM, \bm{M'},\bmP, \bm{P'}),
\end{eqnarray*}
where $\bmM$, $\bm{M'}$, $\bm{M''}$ stand, respectively, for $M^\fraka{}_a$,
$M^\fraka{}_{ab}$ and $M^\fraka{}_{abc}$. Again, \Cref{Lemma:ZQBoundaryData} and
\Cref{Remark:ZQInitialData} allow us to establish the existence and uniqueness of the
trivial solution for this system, which proves that the Yang-Mills equations
are satisfied.

\subsubsection{Summary}

Next, we summarise the above discussion:

\begin{lemma}
\label{Lemma:YM}
Let $F^\fraka{}_{ab}$ and $A^\fraka{}_a$ be fields satisfying the Yang-Mills
equations \eqref{UnphysicalYM1}-\eqref{UnphysicalYM3} with energy-momentum
tensor given by \eqref{EMTensorYM}, and subject to a set of gauge source
functions given by \eqref{Definition:YMGaugeSourceFunction}. If the fields
$f^\fraka{}_i$, $f^{*\fraka}{}_i$, $f^\fraka(x)$, $\mathcal{A}^\fraka$ and
$\notD\mathcal{A}^\fraka$ are prescribed on $\mathscr{I}$ together with values
for $f^\fraka$, $f^{*\fraka}$, $a^\fraka{}_i$ and $a^\fraka$ at
$\partial\mathcal{S}_\star$, then the system
\eqref{ReducedWaveCFE1}-\eqref{ReducedWaveCFE5} coupled to
\eqref{CWEPotentialYM}-\eqref{CWEDerFaradayYM}, written in terms of the reduced
wave operator constitute a proper system of quasilinear wave equations for the
Einstein-Yang-Mills system.
\end{lemma}

\section{Final remarks}
\label{Section:FinalRemarks}

The construction presented in this work depends crucially on the
existence of a quasilinear system of wave equations for the conformal fields
coupled to tracefree matter. Remarkably, the
construction of a homogeneous system of geometric wave equations for the geometric
zero-quantities only assumes a tracefree matter field, without the
specification of the particular model being necessary. Nevertheless, if a
different tracefree matter model is examined, one must proceed to couple it by
constructing suitable quasilinear wave equations as well as initial and
boundary data, not to mention this may require the propagation of a set of
subsidiary variables. The construction of anti-de Sitter-like spacetimes with
an arbitrary matter component under conformal methods still remains as an open
problem. 

Based on the structural properties of the system of wave equations, we expect
that his scheme can be implemented in a straightforward manner in current
numerical codes. In this respect, a key ingredient of the boundary data is the
prescription of the electric part of the Weyl tensor, which is determined
through the solution of a symmetric hyperbolic system. However, neither the
properties of suitable data for this system nor their implications for the
stability of the system are clear. A possible approach consists in exploiting
the theory of symmetric hyperbolic systems to identify the conditions leading
to a well-posed Cauchy problem. This, potentially, might shed further light on
the effect of particular choices of boundary conditions on the
stability/instability of the anti-de Sitter spacetime.

\subsection{Acknowledgments}

DAC thanks support granted by CONACyT (480147), Mexico.



\end{document}